\pdfoutput=1
\RequirePackage{ifpdf}
\ifpdf 
\documentclass[pdftex]{sigma}
\else
\documentclass{sigma}
\fi

\numberwithin{equation}{section}

\newtheorem{Theorem}{Theorem}[section]
\newtheorem{Corollary}[Theorem]{Corollary}
\newtheorem{Lemma}[Theorem]{Lemma}
\newtheorem{Proposition}[Theorem]{Proposition}
 { \theoremstyle{definition}

 }

\usepackage{mathtools} 

\usepackage{tikz} 
	\usetikzlibrary{arrows} 
	\usetikzlibrary{decorations.markings} 
	\usetikzlibrary{patterns} 
	\usetikzlibrary{calc} 

\usepackage[caption=false, font=small]{subfig}

\def\Z{\mathbb{Z}} 
\def\C{\mathbb{C}} 
\def\Im{\textup{Im}}
\def\i{\mathrm{i}}

\begin{document}
\allowdisplaybreaks

\newcommand{\arXivNumber}{2104.04651}

\renewcommand{\PaperNumber}{036}

\FirstPageHeading

\ShortArticleName{A Combinatorial Description of Certain Polynomials Related to the {XYZ} Spin Chain. {II}}

\ArticleName{A Combinatorial Description of Certain Polynomials\\ Related to the {XYZ} Spin Chain.\\ {II}.~The Polynomials~$\boldsymbol{p_n}$}

\Author{Linnea HIETALA~$^{\rm ab}$}

\AuthorNameForHeading{L.~Hietala}

\Address{$^{\rm a)}$~Department of Mathematics, Uppsala University, Box 480, 751~06 Uppsala, Sweden}
\EmailD{\href{mailto:linnea.hietala@math.uu.se}{linnea.hietala@math.uu.se}}

\Address{$^{\rm b)}$~Department of Mathematical Sciences, Chalmers University of Technology\\
\hphantom{$^{\rm b)}$}~and University of Gothenburg, 412~96 Gothenburg, Sweden}

\ArticleDates{Received August 06, 2021, in final form April 29, 2022; Published online May 15, 2022}

\Abstract{By specializing the parameters in the partition function of the 8VSOS model with domain wall boundary conditions and diagonal reflecting end, we find connections between the three-color model and certain polynomials $p_n(z)$, which are conjectured to be equal to certain polynomials of Bazhanov and Mangazeev, appearing in the eigenvectors of the Hamiltonian of the supersymmetric XYZ spin chain. This article is a continuation of a previous paper where we investigated the related polynomials $q_n(z)$, also conjectured to be equal to polynomials of Bazhanov and Mangazeev, appearing in the eigenvectors of the supersymmetric XYZ spin chain.}


\Classification{82B23; 05A15; 33E17}

\tikzset{midarrow/.style={
 decoration={markings,
 mark= at position 0.5 with {\arrow{#1}} ,
 },
 postaction={decorate}
		}
}

\section{Introduction}
In 1967, Lieb \cite{Lieb1967} solved the first special cases of the six-vertex (6V) model on a square lattice, with periodic boundary conditions in both directions. Later the same year, Sutherland \cite{Sutherland1967} solved the general case. A bijection from the states of the 6V model to three-colorings of a~square lattice, where the color is fixed on one face and all adjacent faces have different colors, was found by Lenard \cite[(note added in proof)]{Lieb1967}. The three-color model, where each color is given a weight, was introduced by Baxter \cite{Baxter1970}.

Izergin \cite{Izergin1987, IzerginCokerKorepin1992} obtained a determinant formula for the partition function of the 6V model with domain wall boundary conditions (DWBC) \cite{Korepin1982}. By specializing the parameters in this determinant formula, Kuperberg \cite{Kuperberg1996} gave an alternative proof of the alternating sign matrix (ASM) conjecture of Mills, Robbins and Rumsey \cite{MillsRobbinsRumsey1983}, which gives a formula for the number of ASMs. This conjecture was originally proven a few months earlier by Zeilberger \cite{Zeilberger1996}.
The eight-vertex solid-on-solid (8VSOS) model \cite{Baxter1973} is a two parameter generalization of the 6V model.
Rosengren \cite{Rosengren2011} specialized the parameters of the 8VSOS model with DWBC in Kuperberg's manner and obtained the three-color model.

Kuperberg \cite{Kuperberg2002} later used Tsuchiya's \cite{Tsuchiya1998} determinant formula for the partition function of the 6V model with DWBC and one reflecting end to give a formula for the number of UASMs, which are ASMs with U-turns on one side.
A determinant formula for the partition function of the 8VSOS model with DWBC and one reflecting end was given by Filali \cite{Filali2011} in 2011. 

Razumov and Stroganov \cite{RazumovStroganov2001} found connections between the supersymmetric XXZ spin chain (i.e., where $\Delta=-1/2$ in the Hamiltonian, see, e.g., \cite{HagendorfFendley2012}) and ASMs. Similar problems for the supersymmetric XYZ spin chain were studied by Bazhanov and Mangazeev \cite{BazhanovMangazeev2005, BazhanovMangazeev2010} (see also \cite{RazumovStroganov2010}). The ground state eigenvalues of Baxter's $Q$-operator \cite{Baxter1972} for the eight-vertex (8V) model as well as the components of the ground state eigenvectors of the supersymmetric XYZ-Hamiltonian can be expressed in terms of certain polynomials which seem to have positive integer coefficients \cite{BazhanovMangazeev2006, BazhanovMangazeev2010}. Further investigation has been done by Brasseur and Hagendorf \cite{BrasseurHagendorf2021}.

In Rosengren's study of the 8VSOS model, certain polynomials with positive coefficients showed up and in \cite{Rosengren2015}, Rosengren generalized these polynomials. Zinn-Justin \cite{Zinn-Justin2013} introduced equivalent polynomials and conjectured that Bazhanov's and Mangazeev's polynomials are specializations of these. This led to the suspicion that the combinatorial interpretation of the polynomials of Bazhanov and Mangazeev could have to do with the three-color model.

In this article, as well as in a previous one \cite{Hietala2020}, we assume Zinn-Justin's conjecture, i.e., we use Rosengren's polynomials as the definition of Bazhanov's and Mangazeev's polynomials. In \cite{Hietala2020}, we studied the connection between the three-color model and the polynomials $q_{n-1}(z)$ of Bazhanov and Mangazeev, appearing in the eigenvectors of the Hamiltonian of the supersymmetric XYZ spin chain. By specializing the parameters in the partition function of the 8VSOS model with DWBC and reflecting end in Kuperberg's way, we found an explicit combinatorial expression for $q_{n-1}(z)$, $n\geq 0$, in terms of the partition function of the three-color model on a lattice of size $2n\times n$ with the same boundary conditions.

In \cite{Hietala2020}, we specialized all the spectral parameters in the same way. In the present article, we specialize one of the parameters slightly differently to find a connection between the three-color model on a lattice of size $2n\times n$, and another set of polynomials, which Bazhanov and Mangazeev call $p_{n-1}(z)$, also appearing in the supersymmetric XYZ eigenvectors. Again we consider $n\geq 0$. The expressions that we encounter in this paper get a bit more complicated than in the $q_{n-1}(z)$ case, as the expression includes a sum connected to certain rows of the lattice, which does not appear for $q_{n-1}(z)$.

The paper is organized as follows. In Section~\ref{sec:prel}, we introduce the 8VSOS model as well as the three-color model. In Section~\ref{sec:filali} we specialize the parameters in Filali's determinant formula and write the partition function in terms of certain polynomials, equivalent to $p_n(z)$. In Section~\ref{sec:specializationinpartfcn} we do the same specialization in the very definition of the partition function and get another expression. In Sections~\ref{sec:identificationofterms} and~\ref{sec:pnintermsofT}, we combine the two expressions. Using three-colorings, we finally find combinatorial formulas for $p_n(z)$ in terms of three-colorings.

\vspace{-2mm}

\section{Preliminaries}
\label{sec:prel}
Let $\tau$ and $\eta$ be fixed parameters with $\Im(\tau)>0$ and $\eta\notin \Z+\tau \Z$. Define $p={\rm e}^{2\pi \i\tau}$ and $q={\rm e}^{2\pi \i \eta}$. We then define the theta function
\begin{gather*}\vartheta(x,p)= \prod_{j=0}^{\infty} \big(1-p^j x\big)\big(1-p^{j+1}/x\big).
\end{gather*}
We will often suppress the $p$ and write $\vartheta(x)\coloneqq \vartheta(x, p)$. We write out the second parameter only when it is not just $p$.
We will also use the short hand notation $\vartheta(b x^{\pm a})\coloneqq \vartheta(b x^a)\vartheta(b x^{-a})$.
Then we define
\begin{gather*}
[x]=q^{-x/2}\vartheta\big(q^x, p\big).
\end{gather*}

The most important properties of the theta function are
\begin{gather*}
\vartheta(px)=\vartheta(1/x)=-\frac{1}{x}\vartheta(x),
\end{gather*}
and the addition rule
\begin{gather}
\vartheta(x_1x_3)\vartheta(x_1/x_3)\vartheta(x_2x_4)\vartheta(x_2/x_4) -\vartheta(x_1x_4)\vartheta(x_1/x_4)\vartheta(x_2x_3)\vartheta(x_2/x_3)\nonumber
\\
\qquad
{}=\frac{x_2}{x_3}\vartheta(x_1x_2)\vartheta(x_1/x_2)\vartheta(x_3x_4)\vartheta(x_3/x_4).\label{additionrule}
\end{gather}

{\samepage \noindent
Also define $\omega={\rm e}^{2\pi \i/3}$. Later on we will specify to $\eta=-2/3$ (which corresponds to supersymmetry \cite{HagendorfFendley2012}). In this case $q=\omega$.

}

For any $a\in \Z$ and arbitrary $x$, we have \cite{Rosengren2011}
\begin{align}
&\omega^a=\omega^{a+3}, \qquad 1+\omega+\omega^2=0,\nonumber\\
\label{omegaomega2}
&\vartheta\big(p^{1/2}\omega\big)=\vartheta\big(p^{1/2}\omega^2\big),
\end{align}
and
\begin{gather}\vartheta\big(x\omega^a\big)\vartheta\big(x\omega^{a+1}\big)\vartheta\big(x\omega^{a+2}\big)=\vartheta\big(x^3, p^3\big).
\label{threeproduct}
\end{gather}

\subsection{The 8VSOS model with DWBC and reflecting end}

\begin{figure}[t]
\vspace{3mm}
\centering
\begin{tikzpicture}[scale=0.9,font=\footnotesize]
	\foreach \y in {1,...,3} {
		\draw (.38,1.5*\y-.25-.38) -- +(2.3+0.32,0);
		\draw[midarrow={stealth}] (3,1.5*\y-.25-.38) -- +(1,0) node[right]{$-\lambda_{\y}$};
		\draw (.38,1.5*\y-.25+.38) -- +(2.3+0.32,0);
		\draw[midarrow={stealth},->] (3,1.5*\y-.25+.38) -- +(1,0) node[right]{$\phantom{-}\lambda_{\y}$};
		\draw (0.38,1.5*\y-.25+.38) arc (90:270:0.38);		
	}
	
	\foreach \x in {1,...,3} {
		\draw[midarrow={stealth}] (\x,0) node[below]{$\mu_{\x}$} -- +(0,.87); 
		\draw (\x,.87) -- +(0,3.76); 
		
		\draw[midarrow={stealth reversed}, ->] (\x,4.63) -- +(0,.97);	
	}
	
	\fill[preaction={fill,white},pattern=north east lines, pattern color=gray] (0,0) rectangle (-.15,5.5) ; \draw (0,0) -- (0,5.5);

		\node at (0.5, 5) {$0$};
 \node at (1.5, 5) {$1$};
 \node at (2.5, 5) {$2$};
 \node at (3.65, 5) {$3$};

 \node at (3.65, 4.25) {$2$};
 \node at (3.65, 3.5) {$1$};
 \node at (3.65, 2.75) {$0$};
 \node at (3.5, 2) {$-1$};
 \node at (3.5, 1.25) {$-2$};

 \node at (3.5, 0.5) {$-3$};
 \node at (2.5, 0.5) {$-2$};
 \node at (1.5, 0.5) {$-1$};
 \node at (0.5, 0.5) {$0$};
		
		\node at (0.5, 2) {$0$};
		\node at (0.5, 3.5) {$0$};
\end{tikzpicture}
\vspace{-2mm}
\caption{The 8VSOS model with DWBC and reflecting end in the case $n=3$. The parameters~$\mu_i$ and~$\lambda_i$ are the spectral parameters.}
\label{fig:8vsosdwbcreflend}
\end{figure}
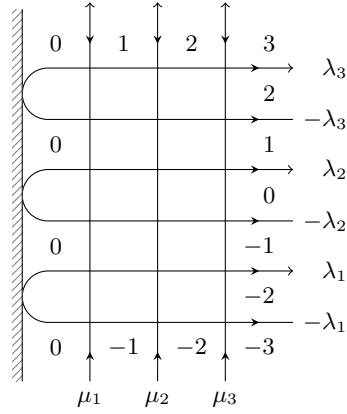

Now we will define the model that we are studying in this paper. Consider a square lattice with $2n\times n$ lines, where the $2n$ horizontal lines are connected pairwise at the left edge as in Figure~\ref{fig:8vsosdwbcreflend}. We equip each line with an orientation. The positive direction is upwards for the vertical lines. For the horizontal double lines, the positive direction goes to the left for the lower part, then it turns and goes to the right on the upper part of the double line. The left side, where the lines turn, is called the reflecting end. To each edge and to each turn we assign a spin $\pm 1$. A lattice with a spin assigned to each edge and each turn is called a state. Graphically we can describe the orientation of a line with a big arrow at the end of the line, and the spins we describe with arrows on the edges. A positive spin, $+1$, corresponds to an arrow pointing in the positive direction of the line, and a negative spin, $-1$, corresponds to an arrow pointing in the opposite direction.

The spins need to follow the so-called ice rule, which demands that at each vertex, two arrows must be pointing inwards and two arrows must be pointing outwards.
Because of the ice rule, there can only be six types of different local configurations of the spins at the vertices. These are the vertices in Figure~\ref{fig:6vmodel}. Because of the reflecting end, for every second row in the square lattice, we need to read off the vertices 90 degrees counterclockwise. We assume two possibilities for the turns, depicted in Figure~\ref{fig:reflectingends}. This corresponds to a diagonal reflection matrix \cite{Filali2011}, and therefore we say that the model has a diagonal reflecting end.
\begin{figure}[t]
\centering
 \subfloat{%
 	\begin{tikzpicture}[scale=0.9, font=\footnotesize]
 	\draw[midarrow={stealth}] (0,1) node[left] {$\lambda$} -- (1,1); 
		\draw[midarrow={stealth reversed}, <-] (2,1) -- (1,1); 
		\draw[midarrow={stealth}] (1,0) node[below] {$\mu$} -- (1,1); 
		\draw[midarrow={stealth reversed}, <-] (1,2) -- (1,1); 
		\draw (1,-1.2) node{\small{$
		a_+(\lambda-\mu, q^{\rho+z})$}};
		\node at (0.4, 1.5) {$z$};
		\node at (1.7, 1.5) {$z-1$};
		\node at (0.3, 0.5) {$z-1$};
		\node at (1.7, 0.5) {$z-2$};
		\end{tikzpicture}
	}\hfil
	\subfloat{%
		\begin{tikzpicture}[scale=0.9, font=\footnotesize]
		\draw[midarrow={stealth}] (0,1) node[left] {$\lambda$} -- (1,1); 
		\draw[midarrow={stealth reversed}, <-] (2,1) -- (1,1); 
		\draw[midarrow={stealth reversed}] (1,0) node[below] {$\mu$} -- (1,1); 
		\draw[midarrow={stealth}, <-] (1,2) -- (1,1); 
		\draw (1,-1.2) node{\small{$
		b_+(\lambda-\mu, q^{\rho+z})$}};
		\node at (0.4, 1.5) {$z$};
		\node at (1.7, 1.5) {$z+1$};
		\node at (0.3, 0.5) {$z-1$};
		\node at (1.6, 0.5) {$z$};
		\end{tikzpicture}
	}\hfil
	\subfloat{%
		\begin{tikzpicture}[scale=0.9, font=\footnotesize]
		\draw[midarrow={stealth}] (0,1) node[left] {$\lambda$} -- (1,1); 
		\draw[midarrow={stealth}, <-] (2,1) -- (1,1); 
		\draw[midarrow={stealth reversed}] (1,0) node[below] {$\mu$} -- (1,1); 
		\draw[midarrow={stealth reversed}, <-] (1,2) -- (1,1); 
		\draw (1,-1.2) node{\small{$
		c_+(\lambda-\mu, q^{\rho+z})$}};
		\node at (0.4, 1.5) {$z$};
		\node at (1.7, 1.5) {$z-1$};
		\node at (0.3, 0.5) {$z-1$};
		\node at (1.6, 0.5) {$z$};
		\end{tikzpicture}
	}\\
	\vspace{0mm}
 \subfloat{%
		\begin{tikzpicture}[scale=0.9, font=\footnotesize]
		\draw[midarrow={stealth reversed}] (0,1) node[left] {$\lambda$} -- (1,1); 
		\draw[midarrow={stealth}, <-] (2,1) -- (1,1); 
		\draw[midarrow={stealth reversed}] (1,0) node[below] {$\mu$} -- (1,1); 
		\draw[midarrow={stealth}, <-] (1,2) -- (1,1); 
		\draw (1,-1.2) node{\small{$
		a_-(\lambda-\mu, q^{\rho+z})$}};
		\node at (0.4, 1.5) {$z$};
		\node at (1.7, 1.5) {$z+1$};
		\node at (0.3, 0.5) {$z+1$};
		\node at (1.7, 0.5) {$z+2$};
		\end{tikzpicture}
	}\hfil
	\subfloat{%
		\begin{tikzpicture}[scale=0.9, font=\footnotesize]
		\draw[midarrow={stealth reversed}] (0,1) node[left] {$\lambda$} -- (1,1); 
		\draw[midarrow={stealth}, <-] (2,1) -- (1,1); 
		\draw[midarrow={stealth}] (1,0) node[below] {$\mu$} -- (1,1); 
		\draw[midarrow={stealth reversed}, <-] (1,2) -- (1,1); 
		\draw (1,-1.2) node{\small{$
		b_-(\lambda-\mu, q^{\rho+z})$}};
		\node at (0.4, 1.5) {$z$};
		\node at (1.7, 1.5) {$z-1$};
		\node at (0.3, 0.5) {$z+1$};
		\node at (1.6, 0.5) {$z$};
		\end{tikzpicture}
	}\hfil
	\subfloat{%
		\begin{tikzpicture}[scale=0.9, font=\footnotesize]
		\draw[midarrow={stealth reversed}] (0,1) node[left] {$\lambda$} -- (1,1); 
		\draw[midarrow={stealth reversed}, <-] (2,1) -- (1,1); 
		\draw[midarrow={stealth}] (1,0) node[below] {$\mu$} -- (1,1); 
		\draw[midarrow={stealth}, <-] (1,2) -- (1,1); 
		\draw (1,-1.2) node{\small{$
		c_-(\lambda-\mu, q^{\rho+z})$}};
		\node at (0.4, 1.5) {$z$};
		\node at (1.7, 1.5) {$z+1$};
		\node at (0.3, 0.5) {$z+1$};
		\node at (1.6, 0.5) {$z$};
		\end{tikzpicture}
	}\\
	\vspace{-1mm}
	\caption{The nonzero vertex weights for the 8VSOS model. The spins are indicated with an arrow halfway the edge, where right and up are positive spins, and left and down are negative spins. The vertex weights depend on the spin configurations, the spectral parameters $\lambda$ and $\mu$, the height $z$ in the upper left face, as well as on the dynamical parameter $\rho$.}
	\label{fig:6vmodel}
\end{figure}
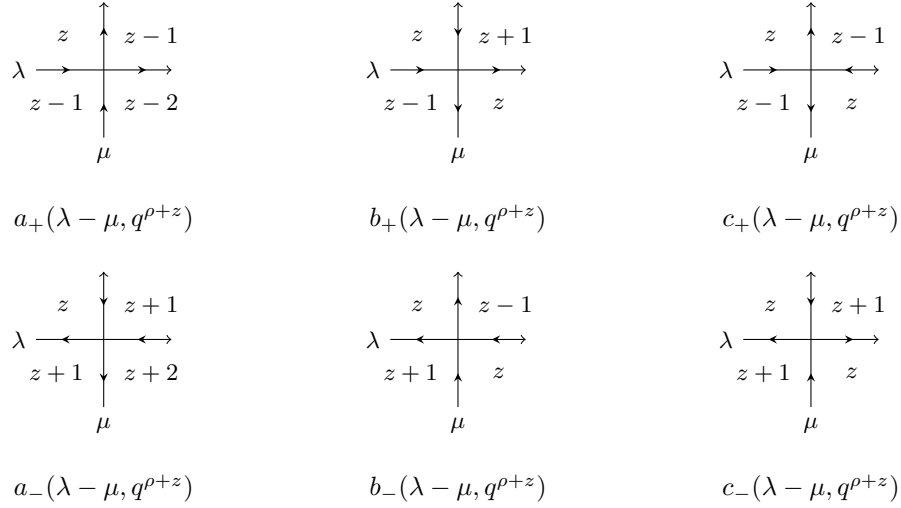

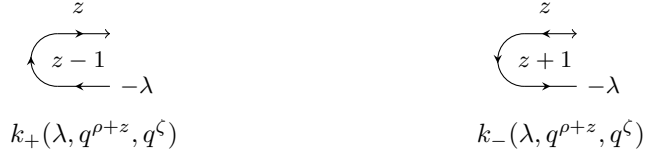
\begin{figure}[t]
\centering
\subfloat{%
 \begin{tikzpicture}[baseline={([yshift=-.5*10pt*0.6]current bounding box.center)}, scale=0.7, font=\footnotesize]
		\draw (0.5,0.5) arc (90:270:0.5);
		\draw[midarrow={stealth reversed}] (.5,-.5) -- (1.5,-.5) node[right] {$-\lambda$};
		\draw[midarrow={stealth}, ->] (.5,+.5) -- (1.5,+.5);
		\draw[-stealth] (0,0.05) -- (0,0.06);
		
	\node at (0.9, 1) {$z$};
	\node at (0.9, 0) {$z-1$};
	
	\draw (1.2,-1.4) node{\small{$
	k_+(\lambda, q^{\rho+z}, q^\zeta)$}};
 \end{tikzpicture}
}\hfil
\subfloat{%
 \begin{tikzpicture}[baseline={([yshift=-.5*10pt*0.6]current bounding box.center)}, scale=0.7, font=\footnotesize]
		\draw (0.5,0.5) arc (90:270:0.5);
		\draw[midarrow={stealth}] (.5,-.5) -- (1.5,-.5) node[right] {$-\lambda$};
	
		\draw[midarrow={stealth reversed},->] (.5,+.5) -- (1.5,+.5);
			\draw[-stealth] (0,-0.05) -- (0,-0.06);
		
	\node at (0.9, 1) {$z$};
	\node at (0.9, 0) {$z+1$};
	
	\draw (1.2,-1.4) node{\small{$
	k_-(\lambda, q^{\rho+z}, q^\zeta)$}};
 \end{tikzpicture}
}
\vspace{-1mm}
\caption{The nonzero boundary weights for the diagonal reflecting end. The weights depend on the spin configurations, the spectral parameter $\lambda$, the height $z$ outside the turn, the dynamical parameter $\rho$, as well as on the boundary parameter $\zeta$.}
 \label{fig:reflectingends}
\end{figure}
To each vertical line we assign a spectral parameter $\mu_j$, and to each horizontal double line we assign a parameter with the value $-\lambda_i$ on the lower part of the line, and the value $\lambda_i$ on the upper part of the line, see Figure~\ref{fig:8vsosdwbcreflend}.
To each face of the lattice, we assign a height $a\in\Z$. Looking in the direction of each spin arrow, the height to the right of the edge must always be 1 smaller than the height to the left. Fixing the height in one single face of a given state determines the heights of the whole state. We fix the height in the upper left corner to be $0$. Finally, we fix a so-called dynamical parameter $\rho\in\C$ and a boundary parameter $\zeta\in\C$.

On the left side of the model we have the reflecting wall. On the remaining three sides, we take DWBC, which means that the spin arrows on the top and the bottom are pointing inwards and the spin arrows on the right boundary are pointing outwards, as in Figure~\ref{fig:8vsosdwbcreflend}. Fixing the height in the upper left corner, all the heights of the faces at the boundaries are determined by the boundary conditions, except for the heights of the faces inside the turns.

\subsection{Weights and partition function of the 8VSOS model}
\label{subsec:8vsosmodel}
To each vertex, we assign a weight which depends on the spins on the adjacent edges, on the height in one of the corners, the dynamical parameter $\rho$, as well as on the spectral parameters $\lambda_i$ and $\mu_j$ on the lines going through the vertex. Each vertex weight is given by a function $w$, which is one of the six functions $a_\pm$, $b_\pm$ or $c_\pm$, depending on the spin configuration as in Figure~\ref{fig:6vmodel}.
To~a~ver\-tex on the upper part of a double line, with spectral parameters $\lambda_i$ and $\mu_j$, and height $z$ in the upper left corner, as in Figure~\ref{fig:6vmodel}, we assign the weight $w(\lambda_i-\mu_j, q^{\rho+z})$. On the lower part of the double line, we need to read off the weights 90 degrees counterclockwise and thus consider the height $z$ in the lower left corner. The weight in this case is $w\big(\mu_j-(-\lambda_i), q^{\rho+z}\big)=w\big(\lambda_i+\mu_j, q^{\rho+z}\big)$. Finally to each turn with spectral parameter $\lambda_i$ and height $z$ inside the turn, as in Figure~\ref{fig:reflectingends}, we assign a weight $w\big(\lambda_i, q^{\rho+z}, q^\zeta\big)$, where $w$ is one of the two functions $k_\pm$ depending on the spins on the turning edge. These local weights are given by
\begin{gather*}
a_+(\lambda, q^z)=a_-(\lambda, q^z)=\frac{[\lambda+1]}{[1]},
\\
b_+(\lambda, q^z)=\frac{[\lambda][z-1]}{[z][1]},\qquad
b_-(\lambda, q^z)=\frac{[\lambda][z+1]}{[z][1]},
\\
c_+(\lambda, q^z)=\frac{[z+\lambda]}{[z]}, \qquad
c_-(\lambda, q^z)=\frac{[z-\lambda]}{[z]},
\\
k_+\big(\lambda, q^z, q^\zeta\big)=\frac{[z+\zeta-\lambda]}{[z+\zeta+\lambda]}, \qquad
k_-\big(\lambda, q^z, q^\zeta\big)=\frac{[\zeta-\lambda]}{[\zeta+\lambda]},
\end{gather*}
and correspond to the different local spin configurations as in Figures~\ref{fig:6vmodel} and~\ref{fig:reflectingends}.
Sometimes when it is the spin configurations around a vertex or at a turn that are of interest, we will refer to an $a_\pm$, $b_\pm$, $c_\pm$ vertex or a $k_+$ (positive) or $k_-$ (negative) turn, without mentioning the different parameters.

The weight of a state is given by the product of all the local weights of the vertices and turns. The partition function is defined as
\begin{align}
\label{firstpartfcn}
Z_n\big(q^{\lambda_1}, \dots, q^{\lambda_n}, q^{\mu_1}, \dots, q^{\mu_n}, q^\rho, q^\zeta\big)=\sum_{\text{states}}\prod_{\text{vertices}} w(\text{vertex}) \prod_{\text{turns}} w(\text{turn}).
\end{align}
The vertex weights in the partition function are either $w(\lambda_i+\mu_j, q^{\rho+z})$ or $w(\lambda_i-\mu_j, q^{\rho+z})$, where $w$ is one of $a_\pm$, $b_\pm$ and $c_\pm$ defined above, and the weights of the turns are one of $k_\pm(\lambda_i, q^{\rho+z}, q^\zeta)$.
It is easy to see that the weights as well as the partition function \eqref{firstpartfcn} are well-defined (see \cite{Hietala2020}).
Filali \cite{Filali2011} used the Izergin--Korepin method to obtain a determinant formula for the partition function of the 8VSOS model with DWBC and reflecting end, namely,
\begin{gather}
Z_n\big(q^{\lambda_1}, \dots, q^{\lambda_n}, q^{\mu_1}, \dots, q^{\mu_n}, q^\rho, q^\zeta\big)\nonumber
\\ \qquad
{}=[1]^{n-2n^2} \prod_{i=1}^n \frac{[2\lambda_i][\zeta-\mu_i][\rho+\zeta+\mu_i][\rho+(2i-n-2)]}{[\zeta+\lambda_i][\rho+\zeta+\lambda_i][\rho+(n-i)]} \nonumber
\\ \qquad\phantom{=}
{}\times\frac{\prod_{i,j=1}^n [\lambda_i+\mu_j+1][\lambda_i-\mu_j+1][\lambda_i+\mu_j][\lambda_i-\mu_j]}{\prod_{1\leq i<j\leq n}[\lambda_i+\lambda_j+1][\lambda_i-\lambda_j][\mu_j+\mu_i][\mu_j-\mu_i]}\det_{1\leq i,j\leq n} K,
\label{Filalisdeterminantformula}
\end{gather}
where
\begin{gather*}
K_{ij}=\frac{1}{[\lambda_i+\mu_j+1][\lambda_i-\mu_j+1][\lambda_i+\mu_j][\lambda_i-\mu_j]}.
\end{gather*}

\subsection{The three-color model and its partition function}

Another model showing up in this paper is the three-color model. This is a model on a square lattice, with the faces filled with three different colors, which we call color $0$, $1$, and $2$, such that adjacent faces have different colors. A weight $t_i$ is assigned to each face of color $i$. A state of the three-color model is called a three-coloring.

We study the three-color model on the $2n\times n$ lattice (i.e., a lattice with $(2n+1)\times(n+1)$ faces).
If we reduce the heights $a$ of the faces in the 8VSOS model modulo $3$, the states of the 8VSOS model can be identified with the states of the three-color model, see Figure~\ref{fig:3colormodel}.
The DWBC and the reflecting end in the 8VSOS model correspond to the following rules for the colors in the three-color model. In the upper left corner, we fix color $0$. On three of the boundaries, the colors alternate cyclically. Starting from the upper left corner, going to the right, the colors increase in the order $0<1<2<0$, to reach $n \operatorname{mod} 3$ in the upper right corner. From there, going down, the colors decrease down to $(-n)\operatorname{mod} 3$ in the lower right corner. Continuing to the left, the colors increase again, up to $0$ in the lower left corner.
On the left hand side, at the reflecting wall, every second face has color $0$. Inside the turns, the colors differ depending on the type of turn in the corresponding state of the 8VSOS model. A negative turn corresponds to color $1$, and a positive turn corresponds to color $2$. We will henceforth assume these boundary conditions, even if we do not mention them explicitly.
The partition function of the three-color model with DWBC and reflecting end, and with color $0$ fixed in the upper left corner, is
\begin{gather*}
Z_n^{3C}(t_0, t_1, t_2)=\sum_{\text{states}} \prod_{\text{faces}} t_i.
\end{gather*}

Let $m$ be the number of positive turns in a state of the 8VSOS model with DWBC and reflecting end. Specifying $m$ means that we have specified the number of faces with color $2$ on the left side. If we specify $m=0$, the colors on the left side alternate between color $0$ and~$1$. There is a bijection between the three-colorings with $m=0$ and the vertically symmetric alternating sign matrices of size $(2n+1)\times (2n+1)$ \cite{Kuperberg2002}.

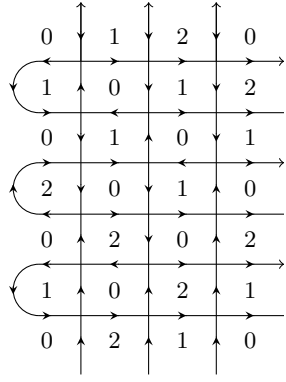
\begin{figure}[t]
\vspace{3mm}
\centering
\begin{tikzpicture}[scale=0.9, font=\footnotesize]
	\foreach \y in {1,...,3} {
		\draw[midarrow={stealth}] (3.55,1.5*\y-.25-.38) -- +(0.01,0);
		\draw[midarrow={stealth}] (3.55,1.5*\y-.25+.38) -- +(0.01,0);
		
		\draw (1, 1.5*\y-.25-.38) -- +(3, 0);
		\draw[->] (1, 1.5*\y-.25+.38) -- +(3, 0);
		
		\draw (0.38,1.5*\y-.25+.38) arc (90:270:0.38);
		}
		
	\foreach \x in {1,...,3} {
		\draw (\x,0) -- +(0,.87); 
			\draw[midarrow={stealth}] (\x,0.55) -- +(0,0.01);
		\draw[->] (\x,4.63) -- +(0,.87);	
			\draw[midarrow={stealth reversed}] (\x,4.63+.43) -- +(0,0.01);	
			
		\draw (\x,.87) -- +(0,4.63);
	}
	
		\draw[-stealth reversed] (0,1.55-0.25) -- (0,1.56-0.25);
		\draw(.38,1.25-.38) -- (1,1.25-.38);
			\draw[midarrow={stealth}] (.55,1.25-.38) -- +(0.01,0);
		\draw(.38,1.25+.38) -- (1,1.25+.38);
			\draw[midarrow={stealth reversed}] (.55,1.25+.38) -- +(0.01,0);
		\node at (0.5, 1.25) {$1$};
		
		\draw[-stealth] (0,3.05-0.25) -- (0,3.06-0.25);
		\draw (.38,2.75-.38) -- (1,2.75-.38);
			\draw[midarrow={stealth reversed}] (.55,2.75-.38) -- +(0.01,0);
		\draw(.38,2.75+.38) -- (1,2.75+.38);
			\draw[midarrow={stealth}] (.55,2.75+.38) -- +(0.01,0);
		\node at (0.5, 2.75) {$2$};
		
		\draw[-stealth reversed] (0,4.45-0.25) -- (0,4.56-0.25);
		\draw(.38,4.25-.38) -- (1,4.25-.38);
			\draw[midarrow={stealth}] (.55,4.25-.38) -- +(0.01,0);
		\draw (.38,4.25+.38) -- (1,4.25+.38);
			\draw[midarrow={stealth reversed}] (.55,4.25+.38) -- +(0.01,0);
		\node at (0.5, 4.25) {$1$};
		
		\draw [midarrow={stealth}] (1.55, 1.5-.25-.38) -- +(0.01, 0);
		\draw [midarrow={stealth}] (2.55, 1.5-.25-.38) -- +(0.01, 0);
		\draw [midarrow={stealth reversed}](1.55, 1.5-.25+.38) -- +(0.01, 0);
		\draw [midarrow={stealth}] (2.55, 1.5-.25+.38) -- +(0.01, 0);
		
		\draw [midarrow={stealth}] (1.55, 3-.25-.38) -- +(0.01, 0);
		\draw [midarrow={stealth}] (2.55, 3-.25-.38) -- +(0.01, 0);
		\draw [midarrow={stealth}] (1.55, 3-.25+.38) -- +(0.01, 0);
		\draw [midarrow={stealth reversed}](2.55, 3-.25+.38) -- +(0.01, 0);
		
		\draw [midarrow={stealth reversed}] (1.55, 4.5-.25-.38) -- +(0.01, 0);
		\draw [midarrow={stealth}] (2.55, 4.5-.25-.38) -- +(0.01, 0);
		\draw [midarrow={stealth}](1.55, 4.5-.25+.38) -- +(0.01, 0);
		\draw [midarrow={stealth}] (2.55, 4.5-.25+.38) -- +(0.01, 0);
		
			\draw [midarrow={stealth}] (1,0.87+.43) -- +(0,0.01);
			\draw [midarrow={stealth}] (1,1.63+.43) -- +(0,0.01);
			\draw [midarrow={stealth reversed}] (1,2.37+.43) -- +(0,0.01);
			\draw [midarrow={stealth reversed}] (1,3.13+.43) -- +(0,0.01);
			\draw [midarrow={stealth}] (1,3.87+.43) -- +(0,0.01);
			
			\draw [midarrow={stealth}] (2,0.87+.43) -- +(0,0.01);
			\draw [midarrow={stealth reversed}] (2,1.63+.43) -- +(0,0.01);
			\draw [midarrow={stealth reversed}] (2,2.37+.43) -- +(0,0.01);
			\draw [midarrow={stealth}] (2,3.13+.43) -- +(0,0.01);
			\draw [midarrow={stealth reversed}] (2,3.87+.43) -- +(0,0.01);
			
			\draw [midarrow={stealth}] (3,0.87+.43) -- +(0,0.01);
			\draw [midarrow={stealth}] (3,1.63+.43) -- +(0,0.01);
			\draw [midarrow={stealth}] (3,2.37+.43) -- +(0,0.01);
			\draw [midarrow={stealth reversed}] (3,3.13+.43) -- +(0,0.01);
			\draw [midarrow={stealth reversed}] (3,3.87+.43) -- +(0,0.01);
			
		\node at (0.5, 5) {$0$};
 \node at (1.5, 5) {$1$};
 \node at (2.5, 5) {$2$};
 \node at (3.5, 5) {$0$};

 \node at (3.5, 4.25) {$2$};
 \node at (3.5, 3.5) {$1$};
 \node at (3.5, 2.75) {$0$};
 \node at (3.5, 2) {$2$};
 \node at (3.5, 1.25) {$1$};

 \node at (3.5, 0.5) {$0$};
 \node at (2.5, 0.5) {$1$};
 \node at (1.5, 0.5) {$2$};
 \node at (0.5, 0.5) {$0$};
		
		\node at (0.5, 2) {$0$};
		\node at (0.5, 3.5) {$0$};
		
		\node at (1.5, 4.25) {$0$};
 \node at (1.5, 3.5) {$1$};
 \node at (1.5, 2.75) {$0$};
 \node at (1.5, 2) {$2$};
 \node at (1.5, 1.25) {$0$};
		
		\node at (2.5, 4.25) {$1$};
 \node at (2.5, 3.5) {$0$};
 \node at (2.5, 2.75) {$1$};
 \node at (2.5, 2) {$0$};
 \node at (2.5, 1.25) {$2$};
\end{tikzpicture}
\vspace{-1mm}
\caption{A state of the three-color model, for $n=3$, with colors $0$, $1$ and $2$. The arrows on the edges show the corresponding state in the 8VSOS model. Here $m=1$.}
\label{fig:3colormodel}
\end{figure}

\section{Connection to Rosengren's polynomials}
\label{sec:filali}
The goal in this paper is to express the partition function in terms of certain polynomials $p_n$, see further Section~\ref{sec:pnintermsofT}. To do this we specialize the variables in the partition function of the 8VSOS model with DWBC and reflecting end.

First, introduce the variables \cite{Rosengren2015}
\begin{gather*}
\psi\coloneqq \psi(\tau)=\frac{\omega^2\vartheta(-1)\vartheta\big({-}p^{1/2}\omega\big)}{\vartheta\big({-}p^{1/2}\big)\vartheta(-\omega)},
\qquad
x(z)=\frac{\vartheta\big({-}p^{1/2}\omega\big)^2\vartheta\big(\omega {\rm e}^{\pm 2\pi \i z}\big)}{\vartheta(-\omega)^2 \vartheta\big(p^{1/2} \omega {\rm e}^{\pm 2\pi \i z}\big)},
\end{gather*}
and define
\begin{gather}
\label{T}
T(x_1, \dots, x_{2n})=\frac{\prod_{i,j=1}^n G(x_j,x_{n+i})}{\Delta(x_1, \dots, x_n)\Delta(x_{n+1}, \dots, x_{2n})}\det_{1\leq i,j\leq n}\bigg(\frac{1}{G(x_j, x_{n+i})}\bigg),
\end{gather}
where $\Delta(x_1, \dots, x_n)=\prod_{1\leq i<j\leq n} (x_j-x_i)$ is the Vandermonde polynomial, and
\begin{gather*}
G(x,y)=(\psi+2)xy(x+y)+\psi(2\psi+1)(x+y)-2\big(\psi^2+3\psi+1\big)xy-\psi\big(x^2+y^2\big).
\end{gather*}
$T$ is a symmetric polynomial \cite{Rosengren2014-1} as well. For $\psi$, the following identities hold \cite[Lemma 9.1]{Rosengren2011}:
\begin{align}
&2\psi+1=\frac{\vartheta\big({-}p^{1/2}\omega\big)^2\vartheta(\omega)^2}{\vartheta(-\omega)^2\vartheta\big(p^{1/2}\omega\big)^2}, \label{lemma9.1-2psi+1}\\
&\psi+1=-\frac{\vartheta\big(p^{1/2}\big)\vartheta\big({-}p^{1/2}\omega\big)}{\vartheta\big({-}p^{1/2}\big)\vartheta\big(p^{1/2}\omega\big)}.\label{lemma9.1-psi+1}
\end{align}
Another identity we will need, which follows from the addition rules~\eqref{additionrule} and \eqref{omegaomega2}, is
\begin{gather}
\label{xzminusxw}
x(z)-x(w)
=\frac{\vartheta\big({-}p^{1/2}\omega\big)^2\vartheta\big(p^{1/2}\omega\big)\vartheta\big(p^{1/2}\big)\omega}{\vartheta(-\omega)^2}\frac{{\rm e}^{-2\pi \i w}\vartheta\big({\rm e}^{2\pi \i(w\pm z)}\big)}{\vartheta\big(p^{1/2} \omega {\rm e}^{\pm 2\pi \i z}\big)\vartheta\big(p^{1/2} \omega {\rm e}^{\pm 2\pi \i w}\big)}.
\end{gather}
We also need the following lemma, which is Lemma~4.1 of \cite{Hietala2020}, where it is proven.
\begin{Lemma}
\label{lemma3.1}
We have
\begin{gather*}
\frac{\vartheta\big({\rm e}^{2\pi \i (w\pm z)}\big)}{\vartheta\big({\rm e}^{6\pi \i (w\pm z)}, p^3\big)}
=\frac{\tilde C {\rm e}^{-4\pi \i w}}{\vartheta\big(p^{1/2}\omega {\rm e}^{\pm 2\pi \i w}\big)^2\vartheta\big(p^{1/2}\omega {\rm e}^{\pm 2\pi \i z}\big)^2} \frac{1}{G(x(z),x(w))},
\end{gather*}
with
\begin{gather*}
\tilde C=\frac{\omega^2\vartheta(-1)\vartheta\big(p^{1/2}\big)^3\vartheta\big(p^{1/2}\omega\big)^2\vartheta\big({-}p^{1/2}\omega\big)^6} {\vartheta(-\omega)^4\vartheta\big({-}p^{1/2}\big)}.
\end{gather*}
\end{Lemma}
As in \cite{Hietala2020}, we rewrite Filali's determinant formula \eqref{Filalisdeterminantformula} and specialize to $\eta=-2/3$.
Let $z_{n+i}=-2(\lambda_i-1)/3$ and $z_j=-2\mu_j/3$ for all $0\leq i,j\leq n$.
By using \eqref{threeproduct}, \eqref{xzminusxw} and Lemma~\ref{lemma3.1} we can rewrite Filali's determinant formula in terms of the polynomials $T$. Details can be found in \cite[Section~4]{Hietala2020}. The partition function becomes
\begin{gather*}
Z_n\big(q^{\lambda_1}, \dots, q^{\lambda_n}, q^{\mu_1}, \dots, q^{\mu_n}, \rho, \zeta\big)
\\ \qquad
{}=(-1)^{\binom{n}{2}}\omega^{\binom{n+1}{2}}\vartheta(\omega)^{n-2n^2} C^{n^2-n} \tilde B\prod_{i=1}^{2n} \vartheta\big(p^{1/2}\omega {\rm e}^{\pm 2\pi \i z_i}\big)^{n-1}
\\ \qquad \hphantom{=}
{}\times\prod_{i=1}^n \frac{\vartheta\big(\omega {\rm e}^{4\pi \i z_{n+i}}\big)\vartheta\big(\zeta {\rm e}^{-2\pi \i z_i}\big)\vartheta\big(\rho\zeta {\rm e}^{2\pi \i z_i}\big)}{\vartheta\big(\zeta \omega {\rm e}^{2\pi \i z_{n+i}}\big)\vartheta\big(\rho\zeta \omega {\rm e}^{2\pi \i z_{n+i}}\big)}\
T(x(z_1), \dots, x(z_{2n})),
\end{gather*}
where
\begin{gather*}
C=\frac{\vartheta(-\omega)^2\vartheta\big({-}p^{1/2}\big)}{\omega\vartheta(-1)\vartheta\big(p^{1/2}\big)^2\vartheta\big(p^{1/2}\omega\big) \vartheta\big({-}p^{1/2}\omega\big)^4}
\end{gather*}
and
\begin{gather*}
\tilde B=
\begin{cases}
1, &\text{for \ $n\equiv 0, 2 \mod 3$},\\[0.5ex]
\dfrac{\vartheta\big(\rho \omega^2\big)}{\vartheta(\rho)}, &\text{for \ $n\equiv 1 \mod 3$}.
\end{cases}
\end{gather*}

The goal is to express the partition function in terms of the polynomials $p_n$ which correspond to $T(2\psi+1, \dots, 2 \psi+1, \psi)$, see further Section~\ref{sec:pnintermsofT}. To do this,
specialize $\lambda_i=1$ for all $i$, and $\mu_j=0$ for $1\leq j\leq n-1$, i.e., $z_i=0$ for all $i$ except for $i=n$.
Recall that $x(0)=2\psi+1$ \eqref{lemma9.1-2psi+1}. \pagebreak

 To get $x(z_n)=\psi$, we must have $z_n=-1/6$ and $\mu_n=1/4$, and hence $q^{\mu_n}=-\omega$.
Specializing to $\mu_n=1/4$ yields
\begin{gather}
Z_n(\underbrace{\omega, \dots, \omega}_{n \text{ times}}, \underbrace{1, \dots, 1}_{n-1 \text{ times}}, -\omega, \rho, \zeta)=(-1)^{\binom{n+1}{2}}\omega^{\binom{n}{2}}\bigg(\frac{C}{\vartheta(\omega)^2}\bigg)^{n^2-n} \nonumber
\\ \qquad
{}\times \tilde B \vartheta\big(p^{1/2}\omega\big)^{2(n-1)(2n-1)}
\big(\vartheta\big({-}p^{1/2}\omega^2\big)\vartheta\big({-}p^{1/2}\big)\big)^{n-1}\nonumber
\\ \qquad
{}\times\bigg(\frac{\vartheta(\zeta)\vartheta(\rho\zeta)}{\vartheta(\zeta \omega)\vartheta(\rho\zeta \omega)}\bigg)^{n-1}\frac{\vartheta\big({-}\zeta \omega^2\big)\vartheta(-\rho\zeta \omega)}{\vartheta(\zeta \omega)\vartheta(\rho\zeta \omega)} T(2\psi+1, \dots, 2\psi+1, \psi).
\label{specializeddet}
\end{gather}

Now we want to rewrite $\big(\frac{\vartheta(\zeta)\vartheta(\rho\zeta)}{\vartheta(\zeta \omega)\vartheta(\rho\zeta \omega)}\big)^{n-1}\frac{\vartheta(-\zeta \omega^2)\vartheta(-\rho\zeta \omega)}{\vartheta(\zeta \omega)\vartheta(\rho\zeta \omega)}$ in terms of $\frac{\vartheta(\rho \zeta \omega^2)}{\vartheta(\rho \zeta \omega)}$ and $\frac{\vartheta(\zeta \omega^2)}{\vartheta( \zeta \omega)}$.
As in~\cite{Hietala2020},
\begin{gather*}
\frac{\vartheta(\zeta)\vartheta(\rho\zeta)}{\vartheta(\zeta \omega)\vartheta(\rho\zeta \omega)}
=-\frac{\omega\vartheta(\rho \omega)}{\vartheta(\rho)}\frac{\vartheta\big(\rho \zeta \omega^2\big)}{\vartheta(\rho \zeta \omega)} -\frac{\omega^2\vartheta\big(\rho \omega^2\big)}{\vartheta(\rho)} \frac{\vartheta\big(\zeta \omega^2\big)}{\vartheta( \zeta \omega)}
\end{gather*}
and
\begin{gather*}
\frac{\vartheta\big({-}\zeta\omega^2\big)\vartheta(-\rho\zeta\omega)}{\vartheta(\zeta \omega)\vartheta(\rho\zeta \omega)}
=-\frac{\omega\vartheta(-1)\vartheta\big({-}\rho\omega^2\big)}{\vartheta(\rho)\vartheta(\omega)}\frac{\vartheta\big(\rho \zeta \omega^2\big)}{\vartheta(\rho \zeta \omega)} +\frac{\vartheta(-\omega)\vartheta(\rho)}{\vartheta(\rho)\vartheta(\omega)} \frac{\vartheta\big(\zeta \omega^2\big)}{\vartheta( \zeta \omega)},
\end{gather*}
because of the addition rule \eqref{additionrule}. Using the binomial theorem we conclude that
\begin{gather*}
\bigg(\frac{\vartheta(\zeta)\vartheta(\rho\zeta)}{\vartheta(\zeta \omega)\vartheta(\rho\zeta \omega)}\bigg)^{n-1}\frac{\vartheta\big({-}\zeta \omega^2\big)\vartheta(-\rho\zeta \omega)}{\vartheta(\zeta \omega)\vartheta(\rho\zeta \omega)}\nonumber
\\ \qquad
{}=\frac{(-1)^n}{\vartheta(\rho)^n\vartheta(\omega)}\sum_{m=0}^n \bigg(\binom{n-1}{m-1}\omega^{2n-m}\vartheta(-1)\vartheta\big({-}\rho\omega^2\big)\vartheta(\rho\omega)^{m-1} \vartheta\big(\rho\omega^2\big)^{n-m}
\\ \qquad \hphantom{=} -\binom{n-1}{m}\omega^{2n-m-2}\vartheta(-\omega)\vartheta(-\rho)\vartheta(\rho\omega)^m \vartheta\big(\rho\omega^2\big)^{n-m-1}\bigg)
\\ \qquad \hphantom{=}
{}\times\bigg(\frac{\vartheta\big(\rho \zeta \omega^2\big)}{\vartheta(\rho \zeta \omega)}\bigg)^m\bigg(\frac{\vartheta\big(\zeta \omega^2\big)}{\vartheta( \zeta \omega)}\bigg)^{n-m}.
\end{gather*}
Inserting this into \eqref{specializeddet} yields
\begin{gather}
Z_n(\omega, \dots, \omega, 1, \dots, 1, -\omega, \rho, \zeta)\nonumber
\\ \qquad
{}=(-\omega)^{\binom{n+1}{2}+n} \bigg(\frac{C}{\vartheta(\omega)^2}\bigg)^{n^2-n}T(2\psi+1, \dots, 2\psi+1, \psi)\nonumber
\\ \qquad\hphantom{=}
{}\times\frac{\tilde B \vartheta\big(p^{1/2}\omega\big)^{2(n-1)(2n-1)}
(\vartheta\big({-}p^{1/2}\omega^2\big)\vartheta\big({-}p^{1/2})\big)^{n-1}}{\vartheta(\rho)^n\vartheta(\omega)} \nonumber
\\ \qquad\hphantom{=}
{}\times\sum_{m=0}^n\bigg(\binom{n-1}{m-1}\omega^{-m}\vartheta(-1)\vartheta\big({-}\rho\omega^2\big) \vartheta(\rho\omega)^{m-1}\vartheta\big(\rho\omega^2\big)^{n-m} \nonumber
-\binom{n-1}{m}\omega^{-m-2}
\\ \qquad\hphantom{=\times}
{}\times
\vartheta(-\omega)\vartheta(-\rho)\vartheta(\rho\omega)^m \vartheta\big(\rho\omega^2\big)^{n-m-1} \bigg)
\bigg(\frac{\vartheta\big(\rho \zeta \omega^2\big)}{\vartheta(\rho \zeta \omega)}\bigg)^m\bigg(\frac{\vartheta\big(\zeta \omega^2\big)}{\vartheta( \zeta \omega)}\bigg)^{n-m}.
\label{partitionfunction2}
\end{gather}
In Section~\ref{sec:pnintermsofT}, we will identify this expression for the partition function with another expression, which we will derive in the next section.

\section{The partition function}
\label{sec:specializationinpartfcn}
In the previous section we were able to write Filali's determinant formula of the partition function in terms of the polynomials $T$. In this section we instead start with the
partition function written in the form
\begin{gather*}
\sum_{\text{states}}\prod_{\text{vertices}} w(\text{vertex}) \prod_{\text{turns}} w(\text{turn}).
\end{gather*}
We specialize the variables in the same way as in the last section to get an expression in terms of a special case of the three-color model.
We follow the same steps as in \cite[Section~3]{Hietala2020}, but in the present paper,
we specialize one of the variables slightly differently, which forces us to treat one of the columns separately through all computations.

Specialize $\lambda_i=1$ and $\mu_j=0$ for all $1\leq i\leq n$ and $1\leq j\leq n-1$ in the partition function.
Everywhere, except at the vertices depending on $\mu_n$, the weights are thus $w(1, \rho q^a)$, where $a$ is the height of the face in the upper left or lower left corner as explained in Section~\ref{subsec:8vsosmodel}.
On the rightmost column of vertices, each vertex weight is either $w(1+ \mu_n, \rho q^a)$ or $w(1- \mu_n, \rho q^a)$. The weights of the turns are $w(1, \rho, \zeta)$. The partition function can hence be written\vspace{-1ex}
\begin{gather*}
Z_n(q, \dots, q, 1, \dots, 1, q^{\mu_n}, \rho, \zeta)
\\ \qquad
{}=\sum_{\text{states}}\Bigg(\prod_{\text{turns}} w(1, \rho, \zeta)\prod_{\substack{\text{vertices not in the} \\ \text{rightmost column}}} w(1, \rho q^a) \prod_{\substack{\text{vertices in the}\\ \text{rightmost column}}} w(1\pm \mu_n, \rho q^a)\Bigg).
\end{gather*}
We first focus on the weights not depending on $\mu_n$.
As in \cite{Hietala2020}, we would like to factor things out from the weights, so that we can write all weights with one single formula, depending only on the heights of the adjacent faces of each vertex, as opposed to the original weights which have different formulas depending on the spins around the vertex. The following lemma is similar to \cite[Proposition~3.2]{Hietala2020} or \cite[Proposition 7.1]{Rosengren2009}.
\begin{Lemma}
\label{lemmaverticeswithoutmun}
Let $\lambda_i=1$ and $\mu_j=0$ for all $i$ and $1\leq j\leq n-1$. For each vertex, let $a$, $b$, $c$, $d$ denote the heights on the adjacent faces as in Figure~$\ref{fig:faceheights}$, and for each turn, let $a$ be the height inside the turn. Let $\nu(w)$ be the number of vertices of type $w$ in a state, and let $\nu_r(w)$ be the number of vertices of type $w$ in the rightmost vertex column. For each state,\vspace{-1ex}
\begin{gather*}
\quad\quad \ \prod_{\mathclap{\substack{\textup{vertices not in the} \\ \textup{rightmost column}}}} w(1, \rho q^a)\prod_{\textup{turns}} w(1, \rho, \zeta)
=P \prod_{\substack{\textup{vertices not in the} \\ \textup{rightmost column}}} \frac{\vartheta\big(\rho q^{\frac{3}{2}a-b+\frac{1}{2}d}\big)}{\vartheta(\rho q^a)}\prod_{\textup{turns}} \frac{\vartheta\big(\rho^{(1-a)/2} \zeta q^{-1}\big)}{\vartheta\big(\rho^{(1-a)/2} \zeta q\big)},
\end{gather*}
where
\begin{gather*}
P=q^{-n(n-1)+\nu(b_+)-\nu_r(b_+)+\nu(c_-)-\nu_r(c_-)+n} \bigg(\frac{\vartheta\big(q^2\big)}{\vartheta(q)}\bigg)^{2n(n-1)-\nu(b_+,b_-,c_+,c_-)+\nu_r(b_+,b_-,c_+,c_-)}.
\end{gather*}
\end{Lemma}

\begin{proof}
We have
\begin{gather*}\vspace{-1ex}
a_+(1, \rho q^a)=a_-(1, \rho q^a)=\frac{q^{-1/2} \vartheta\big(q^2\big)}{\vartheta(q)},
\\
b_+(1, \rho q^a)=\frac{q^{1/2}\vartheta\big(\rho q^{a-1}\big)}{\vartheta(\rho q^a)}, \qquad
b_-(1, \rho q^a)=\frac{q^{-1/2}\vartheta\big(\rho q^{a+1}\big)}{\vartheta(\rho q^a)},
\\
c_+(1, \rho q^a)=\frac{q^{-1/2}\vartheta\big(\rho q^{a+1}\big)}{\vartheta(\rho q^a)}, \qquad
c_-(1, \rho q^a)=\frac{q^{1/2}\vartheta\big(\rho q^{a-1}\big)}{\vartheta(\rho q^a)},
\\
k_+(1, \rho,\zeta)=\frac{q\vartheta\big(\rho \zeta q^{-1}\big)}{\vartheta(\rho \zeta q)},\qquad
 k_-(1, \rho, \zeta)=\frac{q\vartheta\big(\zeta q^{-1}\big)}{\vartheta(\zeta q)}.
\end{gather*}
From each vertex weight $w(1, \rho q^a)$, factor out $q^{-1/2}$, and put it in a prefactor $P$. Then factor out $\frac{\vartheta(q^2)}{\vartheta(q)}$ from the weights of type $a_+$ and $a_-$,
and factor out $q$ from $b_+$ and $c_-$. From each weight of a turn, factor out $q$.
Then we get new vertex weights
\begin{gather*}
\tilde a_+(1, \rho q^a)=\tilde a_-(1, \rho q^a)=1,
\\
\tilde b_+(1, \rho q^a)=\tilde c_-(1, \rho q^a)
=\frac{\vartheta\big(\rho q^{a-1}\big)}{\vartheta(\rho q^a)},\qquad
\tilde b_-(1, \rho q^a)=\tilde c_+(1, \rho q^a)=
\frac{\vartheta\big(\rho q^{a+1}\big)}{\vartheta(\rho q^a)},
\\
\tilde k_+(1, \rho,\zeta)=\frac{\vartheta\big(\rho \zeta q^{-1}\big)}{\vartheta(\rho \zeta q)}, \qquad
\tilde k_-(1, \rho, \zeta)=\frac{\vartheta\big(\zeta q^{-1}\big)}{\vartheta(\zeta q)},
\end{gather*}
and the prefactor is
\begin{gather*}
P=\big(q^{-1/2}\big)^{2n(n-1)} \bigg(\frac{\vartheta\big(q^2\big)}{\vartheta(q)}\bigg)^{2n(n-1)-\nu(b_+,b_-,c_+,c_-)+\nu_r(b_+,b_-,c_+,c_-)}
\!\!\!q^{\nu(b_+)-\nu_r(b_+)+\nu(c_-)-\nu_r(c_-)+n}.
\end{gather*}
One can now verify that for each vertex of type $\tilde a_\pm$, $\tilde b_\pm$ or $\tilde c_\pm$, we have
\begin{gather*}
\tilde w(1, \rho q^a)= \frac{\vartheta\big(\rho q^{\frac{3}{2}a-b+\frac{1}{2}d}\big)}{\vartheta(\rho q^a)},
\end{gather*}
where $a$, $b$, $c$ and $d$ are the heights on the adjacent faces as in Figure~\ref{fig:faceheights}. Furthermore,
\begin{gather*}\tilde k_\pm(1,\rho,\zeta)=\frac{\vartheta\big(\rho^{(1-a)/2} \zeta q^{-1}\big)}{\vartheta\big(\rho^{(1-a)/2} \zeta q\big)},\end{gather*}
where $a$ is the height of the face inside the turn.
\end{proof}

\begin{figure}[!t]
\centering
	\subfloat{%
	\begin{tikzpicture}[scale=0.9, font=\footnotesize]
		\draw[->] (0,1) node[left]{$\lambda_i$} -- (2,1);
		\draw[->] (1,0) node[below]{$\mu_j$} -- (1,2);
		\node at (0.5, 1.5) {$a$};
		\node at (1.5, 1.5) {$b$};
		\node at (0.5, 0.5) {$c$};
		\node at (1.5, 0.5) {$d$};
	\end{tikzpicture}
	}\hfil
	\subfloat{%
\begin{tikzpicture}[scale=0.9, font=\footnotesize]
		\draw[<-] (0,1) -- (2,1) node[right]{$-\lambda_i$};
		\draw[->] (1,0) node[below]{$\mu_j$} -- (1,2);
		\node at (0.5, 1.5) {$b$};
		\node at (1.5, 1.5) {$d$};
		\node at (0.5, 0.5) {$a$};
		\node at (1.5, 0.5) {$c$};
	\end{tikzpicture}
}
\vspace{-2mm}
\caption{Vertices with heights $a$, $b$, $c$ and $d$ on the adjacent faces.}
\label{fig:faceheights}
\end{figure}
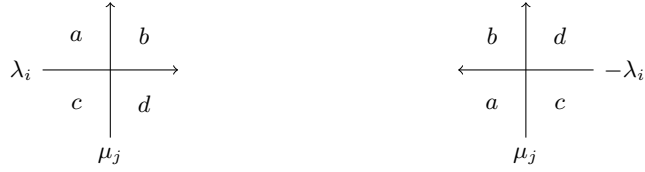

The following lemma is proven in \cite{Hietala2020}.
\begin{Lemma}
\label{numberofnodesoftypew}
For any given state of the $8$VSOS model with DWBC and reflecting end, let $\nu(w)$ be the number of vertices or turns of type $w$. Then we have
\begin{gather*}
\nu(b_+)=\nu(b_-)+\binom{n+1}{2} \qquad \text{and} \qquad \nu(c_+)+2\nu(k_-)=\nu(c_-)+n.
\end{gather*}
\end{Lemma}

\begin{figure}[t]
\centering
\subfloat{%
\begin{tikzpicture}[baseline={([yshift=-.5*10pt*0.6]current bounding box.center)}, scale=0.9, font=\scriptsize]
	\foreach \y in {1,3,5} {
		\draw[midarrow={stealth}] (6.05+0.25,1.5*\y-.25-.38) -- +(0.01,0);
		\draw[midarrow={stealth}] (6.05+0.25,1.5*\y-.25+.38) -- +(0.01,0);
		
		\draw (3, 1.5*\y-.25-.38) -- +(3.5+0.5, 0);
		\draw (0.88+0.25, 1.5*\y-.25-.38) -- +(1.12-0.25, 0);
		\draw[->] (3, 1.5*\y-.25+.38) -- +(3.5+0.5, 0);
		\draw[-] (0.88+0.25, 1.5*\y-.25+.38) -- +(1.12-0.25, 0);
		
		\foreach \x in {-1,...,1} \draw (2.5+.2*\x, 1.5*\y-.25-.38) node{$\cdot\mathstrut$};
		\foreach \x in {-1,...,1} \draw (2.5+.2*\x, 1.5*\y-.25+.38) node{$\cdot\mathstrut$};
		
		\draw (0.88+0.25,1.5*\y-.25+.38) arc (90:270:0.38);
	}
		
	\node[anchor=west] at (6.5+0.5, 1.5-.25-.38) {$-\lambda_1$};
	\node[anchor=west] at (6.5+0.5, 1.5*3-.25-.38) {$-\lambda_k$};
	\node[anchor=west] at (6.5+0.5, 1.5*5-.25-.38) {$-\lambda_n$};
		
	\foreach \x in {1.5,3.5,5.5} {
		\draw (\x,0) -- +(0,3-0.25-.38);
		\draw[midarrow={stealth}] (\x,0.55) -- +(0,0.01);
		\foreach \y in {-1,...,1} \draw (\x, 3-0.25+.2*\y) node{$\cdot\mathstrut$};
		\draw (\x,3-0.25+.38) -- (\x,6-0.25-0.38);
		\foreach \y in {-1,...,1} \draw (\x, 6-0.25+.2*\y) node{$\cdot\mathstrut$};
		\draw (\x,6-0.25+.38) -- (\x,7.5-0.25+0.38);
		\draw[->] (\x,7.5-0.25+0.38) -- +(0,.87);	
		\draw[midarrow={stealth reversed}] (\x,7.5-0.25+0.38+.45) -- +(0,0.01);		
	}
	
	\node at (1.5,-0.3) {${\mu_1}$};
	\node at (3.5,-0.3) {${\mu_{n-1}}$};
	\node at (5.5,-0.3) {${\mu_n}$};

	\foreach \y in {1,5} {
		\draw[midarrow={stealth}] (4.55,1.5*\y-.25-.38) -- +(0.01,0);
		\draw[midarrow={stealth}] (4.55,1.5*\y-.25+.38) -- +(0.01,0);
	}
		\draw[midarrow={stealth reversed}] (4.55,1.5*3-.25-.38) -- +(0.01,0);
		\draw[midarrow={stealth}] (4.55,1.5*3-.25+.38) -- +(0.01,0);
		
	\draw[midarrow={stealth}] (5.5,1.61-0.25) -- +(0,0.01);
	\draw[midarrow={stealth}] (5.5,1.61+0.75-0.25) -- +(0,0.01);
	\draw[midarrow={stealth}] (5.5,3.11-0.75-0.25) -- +(0,0.01);
	\draw[midarrow={stealth}] (5.5,4.61-0.75-0.25) -- +(0,0.01);
	\draw[midarrow={stealth reversed}] (5.5,4.61-0.25) -- +(0,0.01);
	\draw[midarrow={stealth reversed}] (5.5,4.61+0.75-0.25) -- +(0,0.01);
	\draw[midarrow={stealth reversed}] (5.5,7.61-0.75-0.25) -- +(0,0.01);
	\draw[midarrow={stealth reversed}] (5.5,7.61-0.25) -- +(0,0.01);
		
	\fill[preaction={fill,white},pattern=north east lines, pattern color=gray] (0.75,0) rectangle (0.75-.15,8.5) ; \draw (0.75,0) -- (0.75,8.5);
		
		\node[anchor=east] at (1.4, 8) {$0$};
 \node[anchor=west] at (1.5, 8) {$1$};
 \node[anchor=east] at (3.5, 8) {$n-2$};
 \node at (4.5, 8) {$n-1$};
 \node[anchor=west] at (5.5, 8) {$n$};

		\node[anchor=west] at (5.5, 7.25) {$n-1$};
		\node[anchor=west] at (5.5, 6.5) {$n-2$};
 \node[anchor=west] at (5.5, 5) {$-n+2k$};
 \node[anchor=west] at (5.5, 4.25) {$-n+2k-1$};
 \node[anchor=west] at (5.5, 3.5) {$-n+2k-2$};
 \node[anchor=west] at (5.5, 2) {$-n+2$};
 \node[anchor=west] at (5.5, 1.25) {$-n+1$};

 \node[anchor=west] at (5.5, 0.5) {$-n$};
 \node at (4.5, 0.5) {$-n+1$};
 \node[anchor=east] at (3.5, 0.5) {$-n+2$};
 \node[anchor=west] at (1.5, 0.5) {$-1$};
 \node[anchor=east] at (1.4, 0.5) {$0$};
		
		\node[anchor=east] at (1.4, 2) {$0$};
		\node[anchor=east] at (1.4, 3.5) {$0$};
		\node[anchor=east] at (1.4, 5) {$0$};
		\node[anchor=east] at (1.4, 6.5) {$0$};
		\node[anchor=east] at (1.4, 8) {$0$};
		

		\node at (4.5, 7.25) {$n-2$};
		\node at (4.5, 6.5) {$n-3$};
 \node at (4.5, 5) {$-n+2k-1$};
 \node at (4.5, 4.25) {$-n+2k-2$};
 \node at (4.5, 3.5) {$-n+2k-1$};
 \node at (4.5, 2) {$-n+3$};
 \node at (4.5, 1.25) {$-n+2$};
		
\end{tikzpicture}
}\hfil
\subfloat{%
\begin{tikzpicture}[baseline={([yshift=-.5*10pt*0.6]current bounding box.center)}, scale=0.9, font=\scriptsize]
	\foreach \y in {1,3,5} {
		\draw[midarrow={stealth}] (6.05+0.25,1.5*\y-.25-.38) -- +(0.01,0);
		\draw[midarrow={stealth}] (6.05+0.25,1.5*\y-.25+.38) -- +(0.01,0);
		
		\draw (3, 1.5*\y-.25-.38) -- +(3.5+0.5, 0);
		\draw (0.88+0.25, 1.5*\y-.25-.38) -- +(1.12-0.25, 0);
		\draw[->] (3, 1.5*\y-.25+.38) -- +(3.5+0.5, 0);
		\draw[-] (0.88+0.25, 1.5*\y-.25+.38) -- +(1.12-0.25, 0);
		
		\foreach \x in {-1,...,1} \draw (2.5+.2*\x, 1.5*\y-.25-.38) node{$\cdot\mathstrut$};
		\foreach \x in {-1,...,1} \draw (2.5+.2*\x, 1.5*\y-.25+.38) node{$\cdot\mathstrut$};
		
		\draw (0.88+0.25,1.5*\y-.25+.38) arc (90:270:0.38);
	}
		
	\node[anchor=west] at (6.5+0.5, 1.5-.25-.38) {$-\lambda_1$};
	\node[anchor=west] at (6.5+0.5, 1.5*3-.25-.38) {$-\lambda_k$};
	\node[anchor=west] at (6.5+0.5, 1.5*5-.25-.38) {$-\lambda_n$};
		
	\foreach \x in {1.5,3.5,5.5} {
		\draw (\x,0) -- +(0,3-0.25-.38);
		\draw[midarrow={stealth}] (\x,0.55) -- +(0,0.01);
		\foreach \y in {-1,...,1} \draw (\x, 3-0.25+.2*\y) node{$\cdot\mathstrut$};
		\draw (\x,3-0.25+.38) -- (\x,6-0.25-0.38);
		\foreach \y in {-1,...,1} \draw (\x, 6-0.25+.2*\y) node{$\cdot\mathstrut$};
		\draw (\x,6-0.25+.38) -- (\x,7.5-0.25+0.38);
		\draw[->] (\x,7.5-0.25+0.38) -- +(0,.87);	
		\draw[midarrow={stealth reversed}] (\x,7.5-0.25+0.38+.45) -- +(0,0.01);		
	}
	
	\node at (1.5,-0.3) {${\mu_1}$};
	\node at (3.5,-0.3) {${\mu_{n-1}}$};
	\node at (5.5,-0.3) {${\mu_n}$};

	\foreach \y in {1,5} {
		\draw[midarrow={stealth}] (4.55,1.5*\y-.25-.38) -- +(0.01,0);
		\draw[midarrow={stealth}] (4.55,1.5*\y-.25+.38) -- +(0.01,0);
	}
		\draw[midarrow={stealth}] (4.55,1.5*3-.25-.38) -- +(0.01,0);
		\draw[midarrow={stealth reversed}] (4.55,1.5*3-.25+.38) -- +(0.01,0);
		
	\draw[midarrow={stealth}] (5.5,1.61-0.25) -- +(0,0.01);
	\draw[midarrow={stealth}] (5.5,1.61+0.75-0.25) -- +(0,0.01);
	\draw[midarrow={stealth}] (5.5,3.11-0.75-0.25) -- +(0,0.01);
	\draw[midarrow={stealth}] (5.5,4.61-0.75-0.25) -- +(0,0.01);
	\draw[midarrow={stealth}] (5.5,4.61-0.25) -- +(0,0.01);
	\draw[midarrow={stealth reversed}] (5.5,4.61+0.75-0.25) -- +(0,0.01);
	\draw[midarrow={stealth reversed}] (5.5,7.61-0.75-0.25) -- +(0,0.01);
	\draw[midarrow={stealth reversed}] (5.5,7.61-0.25) -- +(0,0.01);
		
	\fill[preaction={fill,white},pattern=north east lines, pattern color=gray] (0.75,0) rectangle (0.75-.15,8.5) ; \draw (0.75,0) -- (0.75,8.5);
		
		\node[anchor=east] at (1.4, 8) {$0$};
 \node[anchor=west] at (1.5, 8) {$1$};
 \node[anchor=east] at (3.5, 8) {$n-2$};
 \node at (4.5, 8) {$n-1$};
 \node[anchor=west] at (5.5, 8) {$n$};

		\node[anchor=west] at (5.5, 7.25) {$n-1$};
		\node[anchor=west] at (5.5, 6.5) {$n-2$};
 \node[anchor=west] at (5.5, 5) {$-n+2k$};
 \node[anchor=west] at (5.5, 4.25) {$-n+2k-1$};
 \node[anchor=west] at (5.5, 3.5) {$-n+2k-2$};
 \node[anchor=west] at (5.5, 2) {$-n+2$};
 \node[anchor=west] at (5.5, 1.25) {$-n+1$};

 \node[anchor=west] at (5.5, 0.5) {$-n$};
 \node at (4.5, 0.5) {$-n+1$};
 \node[anchor=east] at (3.5, 0.5) {$-n+2$};
 \node[anchor=west] at (1.5, 0.5) {$-1$};
 \node[anchor=east] at (1.4, 0.5) {$0$};
		
		\node[anchor=east] at (1.4, 2) {$0$};
		\node[anchor=east] at (1.4, 3.5) {$0$};
		\node[anchor=east] at (1.4, 5) {$0$};
		\node[anchor=east] at (1.4, 6.5) {$0$};
		\node[anchor=east] at (1.4, 8) {$0$};
		

		\node at (4.5, 7.25) {$n-2$};
		\node at (4.5, 6.5) {$n-3$};
 \node at (4.5, 5) {$-n+2k-1$};
 \node at (4.5, 4.25) {$-n+2k$};
 \node at (4.5, 3.5) {$-n+2k-1$};
 \node at (4.5, 2) {$-n+3$};
 \node at (4.5, 1.25) {$-n+2$};
		
\end{tikzpicture}
}
\caption{On the edges between the $(n-1)$th and $n$th column of vertices, there is exactly one left arrow in each state. The single left arrow is on row $l=2k-1$ in the left lattice, and on row $l=2k$ in the right lattice, counted from below.}
\label{fig:leftarrow}
\end{figure}
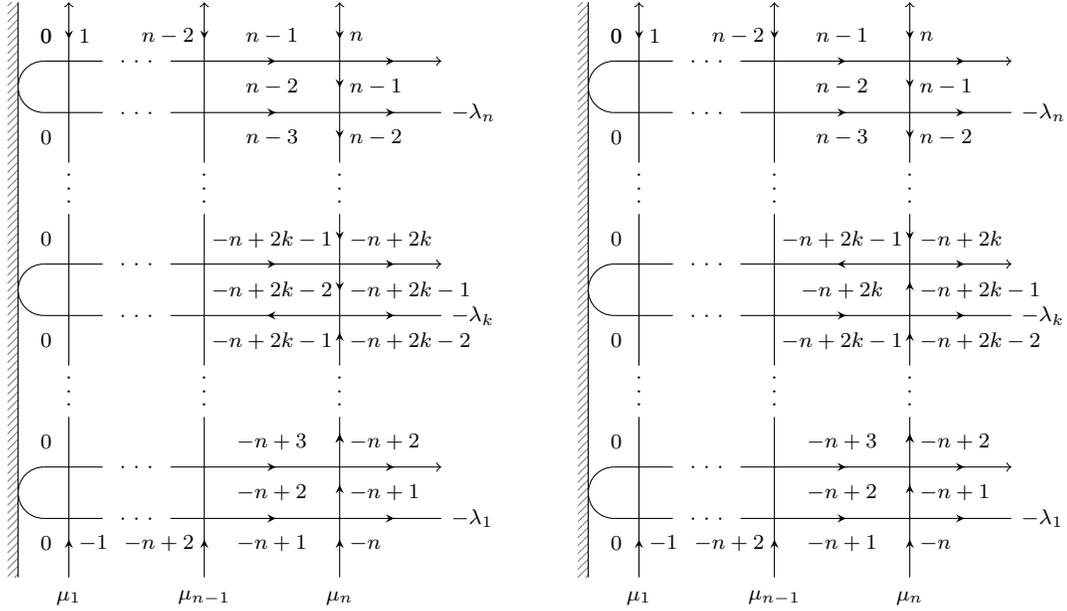

Because of the boundary conditions, we have $2n$ outgoing arrows on the edges to the right, and ingoing arrows on the edges on the top and the bottom. The ice rule thus forces the arrows on the horizontal edges between the $(n-1)$th and the $n$th column (counted from the left) to be right arrows, except for in one place, where the arrow must be a left arrow, see Figure~\ref{fig:leftarrow}.
Given the placement of the left arrow in the second to last column, we know the vertex weights on the entire rightmost column of vertices. We have the following lemma.

\begin{Lemma}
\label{numberofnodesatboundary}
For any given state of the 8VSOS model with DWBC and reflecting end, there is exactly one left arrow on the edges between the $(n-1)$th and the $n$th column of vertices. Assume that the left-pointing arrow is in the $l$th row counted from below. In the case that $l=2k-1$ is odd, the contribution to the partition function from the rightmost column of vertices, i.e., where the weights are depending on $\mu_n$, is
\begin{align*}
\qquad\ \prod_{\mathclap{\substack{\textup{vertices in the} \\ \textup{rightmost column}}}} w(1\pm \mu_n, \rho q^a)={}&\prod_{i=1}^{k-1} \big[ b_+\big(1+\mu_n, \rho q^{-n+2i-1}\big) a_+\big(1-\mu_n, \rho q^{-n+2i+1}\big)\big]
\\
&\times c_+\big(1+\mu_n, \rho q^{-n+2k-1}\big) b_+\big(1-\mu_n, \rho q^{-n+2k-1}\big)
\\
&\times\prod_{i=k+1}^{n} \big[ a_-\big(1+\mu_n, \rho q^{-n+2i-3}\big) b_+\big(1-\mu_n, \rho q^{-n+2i-1}\big)\big].
\end{align*}
In the case that $l=2k$ is even, the contribution is
\begin{align*}
\qquad\ \prod_{\mathclap{\substack{\textup{vertices in the} \\ \textup{rightmost column}}}} w(1\pm \mu_n, \rho q^a)
={}&\prod_{i=1}^{k-1} \big[b_+\big(1+\mu_n, \rho q^{-n+2i-1}\big) a_+\big(1-\mu_n, \rho q^{-n+2i+1}\big)\big]
\\
&\times b_+\big(1+\mu_n, \rho q^{-n+2k-1}\big) c_-\big(1-\mu_n, \rho q^{-n+2k-1}\big)\\
&\times\prod_{i=k+1}^n \big[a_-\big(1+\mu_n, \rho q^{-n+2i-3}\big) b_+\big(1-\mu_n, \rho q^{-n+2i-1}\big)\big].
\end{align*}
\end{Lemma}

Because of Lemma~\ref{numberofnodesoftypew}, we know the difference between the number of $b_+$ and $b_-$ vertices, and the difference between the number of $c_+$ and $c_-$ vertices. Because of Lemma~\ref{numberofnodesatboundary}, we know the number of $b_\pm$ and $c_\pm$ in the rightmost vertex column, and they only depend on whether $l$ is even or odd.

\begin{Corollary}
Let $\nu_r(w)$ be the number of vertices of type $w$ in the rightmost vertex column, and let $l$ be the row where the single left arrow of the rightmost column sits. Then
\begin{gather*}
\nu_r(b_+)=n, \qquad \nu_r(b_-)=0,
\end{gather*}
and
\begin{gather*}\nu_r(c_+)=
\begin{cases}
1, & \text{if\ $l$ is odd,}\\
0, & \text{if\ $l$ is even},
\end{cases}
\qquad
\nu_r(c_-)=
\begin{cases}
0, & \text{if\ $l$ is odd,}\\
1, & \text{if\ $l$ is even}.
\end{cases}
\end{gather*}
\end{Corollary}

With $\eta=-2/3$, the prefactor from Lemma~\ref{lemmaverticeswithoutmun} thus becomes
\begin{gather*}
P
=(-1)^{\binom{n+1}{2}-1}\omega^{-\binom{n+1}{2}-2(n-m)+\alpha},
\end{gather*}
where $\alpha= (-1)^l$ and $m$ is the number of $k_+$ turns.
As in \cite{Hietala2020, Rosengren2009}, specializing $\eta=-2/3$ yields that the weights from Lemma~\ref{lemmaverticeswithoutmun} now become
\begin{gather*}
\tilde w(1, \rho q^a)=\frac{\vartheta\big(\rho q^{\frac{3}{2}a-b+\frac{1}{2}d}\big)}{\vartheta(\rho q^a)}=\frac{\vartheta\big(\rho q^{-b-d}\big)}{\vartheta(\rho q^a)}
=\frac{\vartheta\big(\rho^3, p^3\big)}{\vartheta(\rho \omega^a)\vartheta(\rho \omega^b)\vartheta(\rho \omega^d)},
\end{gather*}
 since $b$, $d$ and $-b-d$ are noncongruent modulo $3$.
Hence for each vertex in the lattice, except for in the rightmost column of vertices, we just need to know the heights of three of the adjacent faces.
As in \cite{Hietala2020, Rosengren2009}, this means that we need to account for each face in the interior three times, but in the present article we need to treat the $n$th column of faces differently. On the boundary, the number of faces that we need to count differs, but here we know the heights, so we can explicitly compute this contribution. Likewise it differs in the $n$th column, but given the number $l$, we can compute this contribution explicitly as well. On the right boundary, we should not count any faces at all, since these faces only have to do with the weights containing~$\mu_n$. The faces of each turn should be accounted for only once.

The expression of Lemma~\ref{lemmaverticeswithoutmun} is now given by
\begin{gather*}
\qquad\ \prod_{\mathclap{\substack{\textup{vertices not in the} \\ \textup{rightmost column}}}} w(1, \rho \omega^a)\prod_{\textup{turns}} w(1, \rho, \zeta)
\\ \qquad
{}=P C \vartheta\big(\rho^3, p^3\big)^{2n(n-1)} \prod_{\textup{faces}} \frac{1}{\vartheta(\rho\omega^a)^3} \prod_{\text{turns}} \vartheta(\rho \omega^a)^2\frac{\vartheta\big(\rho^{(1-a)/2} \zeta \omega^2\big)}{\vartheta\big(\rho^{(1-a)/2} \zeta \omega\big)}
\\ \qquad
{}=P C \vartheta\big(\rho^3, p^3\big)^{2n(n-1)} \vartheta(\rho \omega)^{2(n-m)}\vartheta\big(\rho \omega^2\big)^{2m}\bigg(\!\frac{\vartheta\big(\rho\zeta \omega^2\big)}{\vartheta(\rho \zeta \omega)}\!\bigg)^m\bigg(\!\frac{\vartheta\big(\zeta \omega^2\big)}{\vartheta(\zeta \omega)}\!\bigg)^{n-m} \!\! \prod_{\text{faces}} \!\frac{1}{\vartheta(\rho\omega^a)^3},
\end{gather*}
where $m$ is the number of positive turns and $C$ is the correction compensating for the faces that are counted too many times. The correction $C$ is given by
\begin{align*}
C&=R \vartheta(\rho)^{n+3}\vartheta\big(\rho q^{n-1}\big)^2\vartheta\big(\rho q^{-n+1}\big)^3\prod_{i=1}^{n-2}\vartheta\big(\rho q^i\big) \prod_{i=-n}^n\vartheta\big(\rho q^i\big)^3\prod_{i=1}^{n-2}\vartheta\big(\rho q^{-i}\big)^2
\\
&=\vartheta(\rho)^n\vartheta\big(\rho^3, p^3\big)^{3n} R\times
\begin{cases}
\vartheta(\rho)^3\vartheta(\rho\omega)\vartheta\big(\rho\omega^2\big), & n\equiv 0 \mod 3,\\
\vartheta(\rho)^5, & n\equiv 1 \mod 3,\\
\vartheta(\rho\omega)^2\vartheta\big(\rho\omega^2\big)^3, &n\equiv 2 \mod 3,
\end{cases}
\end{align*}
where $R$ is the contribution to the correction from the second to last column of faces, which depends on $l$.
We have
\begin{align*}
R={}&\prod_{i=1}^{k-1}\big(\vartheta\big(\rho q^{-n+2i}\big)\vartheta\big(\rho q^{-n+2i+1}\big)^2\big)
\prod_{i=k}^{n-1}\big(\vartheta\big(\rho q^{-n+2i-1}\big)^2\vartheta\big(\rho q^{-n+2i}\big)\big)
\\
&\times \begin{cases}
\vartheta\big(\rho q^{-n+2k-2}\big), &\text{for\ $l=2k-1$ odd,}\\
\vartheta\big(\rho q^{-n+2k}\big), &\text{for\ $l=2k$ even}.
\end{cases}
\end{align*}
The expression in Lemma~\ref{numberofnodesatboundary} is\vspace{-1ex}
\begin{align*}
\qquad\ \prod_{\mathclap{\substack{\textup{vertices in the} \\ \textup{rightmost column}}}} w(1\pm \mu_n, \rho q^a)
={}&\frac{\vartheta\big(\rho q^{-n}\big)\vartheta\big(\rho q^{-n+2}\big)\cdots \vartheta\big(\rho q^{n-4}\big)\vartheta\big(\rho q^{n-2}\big)}{\vartheta\big(\rho q^{-n+1}\big)\vartheta\big(\rho q^{-n+3}\big)\cdots\vartheta\big(\rho q^{n-3}\big)\vartheta\big(\rho q^{n-1}\big)}
\\
&\times \frac{\big(\vartheta\big(q^{1+\mu_n}\big)\vartheta\big(q^{2-\mu_n}\big)\big)^{k-1}\big(\vartheta\big(q^{2+\mu_n}\big) \vartheta\big(q^{1-\mu_n}\big)\big)^{n-k}}{\vartheta(q)^{2n-1}\vartheta\big(\rho q^{-n+2k-1}\big)}
\\
&\times\begin{cases}
\vartheta\big(\rho q^{-n+2k+\mu_n}\big)\vartheta\big(q^{1-\mu_n}\big), &\text{for \ $l=2k-1$ odd},\\
q^{1-\mu_n}\vartheta\big(\rho q^{-n+2k-2+\mu_n}\big)\vartheta\big(q^{1+\mu_n}\big), &\text{for \ $l=2k$ even}.
\end{cases}
\end{align*}
Putting all of the above together, we get the following proposition.
\begin{Proposition}\label{partitionfunctionmun}
Let $\eta=-2/3$, $\lambda_i=1$ for all $i$, and $\mu_j=0$ for $1\leq j\leq n-1$. Then the partition function is\vspace{-1ex}
\begin{gather*}
Z_n\big(\omega, \dots, \omega, 1, \dots, 1, q^{\mu_n}, \rho, \zeta\big)
\\ \qquad
{}=\big({-}\omega^2\big)^{\binom{n+1}{2}+1} \frac{\vartheta(\rho)^{n+1}\vartheta\big(\rho^3, p^3\big)^{2n(n+1)-1}}{{\vartheta(\omega)^{2n-1}}} \tilde X
\\ \qquad \hphantom{=}
{}\times\sum_{m=0}^n
\frac{\vartheta(\rho \omega)^{2(n-m)}\vartheta\big(\rho \omega^2\big)^{2m}}{\omega^{2(n-m)}}\bigg(\frac{\vartheta\big(\rho\zeta \omega^2\big)}{\vartheta(\rho \zeta \omega)}\bigg)^m\bigg(\frac{\vartheta\big(\zeta \omega^2\big)}{\vartheta(\zeta \omega)}\bigg)^{n-m}
\\ \qquad \hphantom{=}
{}\times \sum_{l=0}^{2n}
\big(\vartheta\big(q^{1+\mu_n}\big)\vartheta\big(q^{2-\mu_n}\big)\big)^{k-1}\big(\vartheta\big(q^{2+\mu_n}\big)\vartheta\big(q^{1-\mu_n}\big)\big)^{n-k}
\\ \qquad\hphantom{=\times}
{}\times\vartheta\big(\rho q^{-n+l}\big)\vartheta\big(\rho q^{-n+l-1}\big) \vartheta\big(\rho q^{-n+l+1+\mu_n}\big)\tilde Y \sum_{\substack{\textup{states with}\\ \textup{$m$ positive turns}
\\ \textup{and $\leftarrow$ on $l$th row}}}\prod_{\textup{faces}} \frac{1}{\vartheta(\rho\omega^a)^3},
\end{gather*}
where $k$ is defined by $l=2k$ for even $l$, and by $l=2k-1$ for odd $l$,\vspace{-1ex}
\begin{gather*}
\tilde X =\begin{cases}
\vartheta(\rho)^2, & n\equiv 0 \mod 3,\\
\vartheta(\rho)\vartheta\big(\rho\omega^2\big), & n\equiv 1 \mod 3,\\
\vartheta(\rho\omega)\vartheta\big(\rho\omega^2\big), &n\equiv 2 \mod 3,
\end{cases}
\end{gather*}
and\vspace{-1ex}
\begin{gather*}
\tilde Y=\begin{cases}
\vartheta\big(q^{1-\mu_n}\big), &\text{for\ $l=2k-1$ odd},\\
q^{-\mu_n}\vartheta\big(q^{1+\mu_n}\big), &\text{for\ $l=2k$ even}.
\end{cases}
\end{gather*}
\end{Proposition}

In the partition function in Proposition~\ref{partitionfunctionmun}, we now specify $\mu_n=1/4$. Then\vspace{-1ex}
\begin{gather}
Z_n(\omega, \dots, \omega, 1, \dots, 1, -\omega, \rho, \zeta)\nonumber
\\ \qquad
{}=\big({-}\omega^2\big)^{\binom{n+1}{2}+1}\frac{\vartheta(\rho)^{n+3}\vartheta(-1)^{2n-1}\vartheta\big(\rho^3, p^3\big)^{2n(n+1)-1}}{{\vartheta(\omega)^{2n-1}}} X \nonumber
\\ \qquad\hphantom{=}
{}\times\sum_{m=0}^n\frac{\vartheta(\rho \omega)^{2(n-m)}\vartheta\big(\rho \omega^2\big)^{2m}}{\omega^{2(n-m)}}\bigg(\frac{\vartheta\big(\rho\zeta \omega^2\big)}{\vartheta(\rho \zeta \omega)}\bigg)^m\bigg(\frac{\vartheta\big(\zeta \omega^2\big)}{\vartheta(\zeta \omega)}\bigg)^{n-m}\nonumber
\\ \qquad\hphantom{=}
{}\times \sum_{l=0}^{2n} \Biggl((-1)^{l-1}\omega^{l-1}\bigg(\frac{\vartheta(-\omega)}{\vartheta(-1)}\bigg)^{l-1}\vartheta\big(\rho \omega^{-n+l}\big)\vartheta\big(\rho \omega^{-n+l-1}\big) \vartheta\big({-}\rho \omega^{-n+l+2}\big)\nonumber
\\ \qquad\hphantom{=\times}
{}\times\sum_{\substack{\textup{states with}\\ \textup{$m$ positive turns}\\ \textup{and $\leftarrow$ on $l$th row}}}\prod_{\text{faces}} \frac{1}{\vartheta(\rho\omega^a)^3}\Biggr),
\label{partitionfunction3}
\end{gather}

\noindent
for
\begin{gather*}
X =\begin{cases}
1, & n\equiv 0 \mod 3,\\[0.5ex]
\dfrac{\vartheta\big(\rho\omega^2\big)}{\vartheta(\rho)}, & n\equiv 1 \mod 3,\\[2ex]
\dfrac{\vartheta(\rho\omega)\vartheta\big(\rho\omega^2\big)}{\vartheta(\rho)^2}, &n\equiv 2 \mod 3.
\end{cases}
\end{gather*}

The 6V model with DWBC and a reflecting end with the unique left arrow of the second to last column fixed to a given edge $l$ has earlier been studied, e.g., in the context of so-called boundary correlation functions (see, e.g., \cite{Wang2003}). These functions describe the probability that the second to last column of a state has its unique left arrow on the $l$th row. It is also related to refined enumerations of UASMs (see, e.g., \cite{RazumovStroganov2004-2}), where the objective is to count the number of UASMs that have the unique $1$ of the rightmost column on the $l$th row.

\section{Identification of terms}
\label{sec:identificationofterms}
In this section, we combine the two expressions \eqref{partitionfunction2} and \eqref{partitionfunction3} for the partition function and identify the terms with the same $m$. This we can do since $\big(\frac{\vartheta(\rho\zeta \omega^2)}{\vartheta(\rho \zeta \omega)}\big)^m\big(\frac{\vartheta(\zeta \omega^2)}{\vartheta(\zeta \omega)}\big)^{n-m}$ are linearly independent as functions of $\zeta$. Then
\begin{gather}
\frac{(-1)^n E^{n^2-n}D^{n-1} B'}{\vartheta(\rho)^{2n^2+4n+3}\vartheta(\rho\omega)^{2n^2+4n-3m}\vartheta\big(\rho\omega^2\big)^{2n^2+n+3m}} T(2\psi+1, \dots, 2\psi+1, \psi)\nonumber
\\ \qquad\hphantom{=}
{}\times\bigg(\binom{n-1}{m-1}\vartheta(\rho)\vartheta\big(\rho\omega^2\big)\vartheta\big({-}\rho\omega^2\big) -\binom{n-1}{m}\omega\frac{\vartheta(-\omega)}{\vartheta(-1)}\vartheta(\rho) \vartheta(\rho\omega)\vartheta(-\rho) \bigg)\nonumber
\\ \qquad
{}=\sum_{l=0}^{2n} \bigg((-1)^{l}\bigg(\omega\frac{\vartheta(-\omega)}{\vartheta(-1)}\bigg)^{l-1}\vartheta\big(\rho \omega^{-n+l}\big)\vartheta\big(\rho \omega^{-n+l+2}\big) \vartheta\big({-}\rho \omega^{-n+l+2}\big)\nonumber
\\ \qquad\hphantom{=}
{}\times\sum_{\substack{\textup{states with}\\ \textup{$m$ positive turns}\\ \textup{and $\leftarrow$ on $l$th row}}}\prod_{\text{faces}} \frac{1}{\vartheta(\rho\omega^a)^3}\bigg),
\label{finalpartitionfunction}
\end{gather}
with
\begin{gather*}
E=\frac{\vartheta(-\omega)^2\vartheta\big({-}p^{1/2}\big)\vartheta\big(p^{1/2}\omega\big)^3} {\vartheta(\omega)^2\vartheta(-1)\vartheta\big(p^{1/2}\big)^2\vartheta\big({-}p^{1/2}\omega\big)^4}, \qquad D=\frac{\omega^2\vartheta\big({-}p^{1/2}\omega\big)\vartheta\big({-}p^{1/2}\big)\vartheta(\omega)^2} {\vartheta\big(p^{1/2}\omega\big)^2\vartheta(-1)^2},
\end{gather*}
and
\begin{gather*}
B'=
\begin{cases}
1, & n\equiv 0,1 \mod 3,\\[0.5ex]
\dfrac{\vartheta(\rho)^2}{\vartheta(\rho\omega)\vartheta\big(\rho\omega^2\big)}, &n\equiv 2 \mod 3.
\end{cases}
\end{gather*}

Focus on the right hand side of \eqref{finalpartitionfunction}.
Let $N_{m,l}(k_0,k_1,k_2)$ be the number of states with $m$ positive turns, the left arrow between the $(n-1)$th and $n$th column in the $l$th row from below, and with $k_i$ faces of color $i$. Then
\begin{gather*}
\sum_{\substack{\textup{states with}\\ \textup{$m$ positive turns}\\ \textup{and $\leftarrow$ on $l$th row}}}\prod_{\text{faces}} \frac{1}{\vartheta(\rho\omega^a)^3}
=\sum_{(k_0, k_1, k_2)\in \Z^3}N_{m,l}(k_0, k_1, k_2)\frac{1}{\vartheta(\rho)^{3k_0}\vartheta(\rho\omega)^{3k_1}\vartheta\big(\rho\omega^2\big)^{3k_2}}.
\end{gather*}
This is the partition function for the three-color model with fixed values of $m$ and $l$.
That a state has its left arrow on the $l$th row from below means for the three-coloring that when starting from below, the single edge where the color decreases by $1$ ${\rm mod}~3$ is the $l$th edge from below.

\subsection{Specialization of the dynamical parameter}
At this point in the investigation of the polynomials $q_n$ in the previous article \cite{Hietala2020}, we defined the parameters $t_i=1/\vartheta(\rho\omega^i)^3$ 
to be able to study the partition function of the three-color model.
However, due to the summation over $l$ in the present article, the same substitution in \eqref{finalpartitionfunction} leads to an expression with sums, which seems hard to make use of to study the full three-color model partition function. Therefore, we will not use these parameters. Instead, we will consider the special cases $\rho=-\omega^i$, $i=0,1,2$. This allows us to study cases of the three-color model where two colors have equal weight.

Rewrite the sum on the right hand side of \eqref{finalpartitionfunction} as three separate sums for $l\equiv n$, $n+1$, $n+2 \operatorname{mod} 3$ respectively. Define
\begin{gather*}
\xi=\omega\frac{\vartheta(-\omega)}{\vartheta(-1)}.
\end{gather*}
Then \eqref{finalpartitionfunction} becomes
\begin{gather}
\frac{(-1)^n E^{n^2-n}
D^{n-1}B'}{\vartheta(\rho)^{2n^2+4n+3}\vartheta(\rho\omega)^{2n^2+4n-3m} \vartheta\big(\rho\omega^2\big)^{2n^2+n+3m}} T(2\psi+1, \dots, 2\psi+1, \psi)\nonumber
\\ \qquad
{}\times\bigg(\binom{n-1}{m-1}\vartheta(\rho)\vartheta\big(\rho\omega^2\big)\vartheta\big({-}\rho\omega^2\big) -\binom{n-1}{m}\xi\vartheta(\rho)\vartheta(\rho\omega)\vartheta(-\rho) \bigg)\nonumber
\\ \qquad
{}=\sum_{(k_0, k_1, k_2)\in \Z^3}\sum_{l\equiv n \operatorname{mod} 3} \big((-1)^{l}\xi^{l-1}\vartheta(\rho)\vartheta\big(\rho \omega^2\big) \vartheta\big({-}\rho \omega^2\big)N_{m,l}(k_0, k_1, k_2)\nonumber
\\ \qquad\phantom{=}
{}+(-1)^{l}\xi^l\vartheta(\rho \omega)\vartheta(\rho ) \vartheta(-\rho)N_{m,l+1}(k_0, k_1, k_2)\nonumber
\\ \qquad\phantom{=}
{}+(-1)^{l+1}\xi^{l+1}\vartheta\big(\rho \omega^2\big)\vartheta(\rho \omega) \vartheta(-\rho \omega)N_{m,l+2}(k_0, k_1, k_2)\big)\nonumber
\\ \qquad\phantom{=}
{}\times\frac{1}{\vartheta(\rho)^{3k_0}\vartheta(\rho\omega)^{3k_1}\vartheta\big(\rho\omega^2\big)^{3k_2}}.
\label{part2}
\end{gather}
Now specify $\rho=-1$. Terms with $\vartheta(-\rho)$ will then vanish. Divide both sides by $\vartheta(\omega)\vartheta(-1)^2$.
Observe that $k_0+k_1+k_2=(2n+1)(n+1)$ and multiply both sides by $\vartheta(-\omega)^{3(2n+1)(n+1)}$. Then we get
\begin{gather*}
E^{n^2-n}
D^{n-1}\xi^{2n^2+4n+b}\binom{n-1}{m-1}T(2\psi+1, \dots, 2\psi+1, \psi)
\\ \qquad
{}=\sum_{(k_0, k_1, k_2)\in \Z^3}\sum_{l\equiv n \operatorname{mod} 3} (-1)^{n+l}\xi^l
(N_{m,l}(k_0, k_1, k_2)+N_{m,l-1}(k_0, k_1, k_2))\xi^{3k_0},
\end{gather*}
where
\begin{gather*}
 b=
\begin{cases}
4, & n\equiv 0,1 \mod 3,\\
2, &n\equiv 2 \mod 3.
\end{cases}
\end{gather*}
We can thus rewrite the sum to depend only on $k_0$, and not on $k_1$ and $k_2$. We write
\begin{gather}
E^{n^2-n}
D^{n-1}\xi^{2n^2+4n+b}\binom{n-1}{m-1}T(2\psi+1, \dots, 2\psi+1, \psi)\nonumber
\\ \qquad
{}=\sum_{k_0\in\Z}\sum_{l\equiv n \operatorname{mod} 3} (-1)^{n+l}\xi^l(N_{m,l}(k_0)+N_{m,l-1}(k_0))\xi^{3k_0},
\label{partwithrhominus1}
\end{gather}
where $N_{m,l}(k_i)$ is the number of states with $m$ positive turns, the left arrow between the $(n-1)$th and $n$th column in the $l$th row from below, and with $k_i$ faces of color $i$.

As in \cite{Hietala2020},
\begin{gather}\label{eq:n01}
\bigg(\frac{\vartheta(\omega)^2\vartheta(-1)\vartheta\big(p^{1/2}\big)^2\vartheta\big({-}p^{1/2}\omega\big)^4} {\vartheta(-\omega)^2\vartheta\big({-}p^{1/2}\big)\vartheta\big(p^{1/2}\omega\big)^3}\bigg)^6
=2^4\psi^{2}(\psi+1)^8(2\psi+1)^6,
\\[.5ex]
\label{eq:n03}
\frac{1}{\xi^2}\bigg(\frac{\vartheta(\omega)^2\vartheta(-1)\vartheta\big(p^{1/2}\big)^2\vartheta\big({-}p^{1/2}\omega\big)^4} {\vartheta(-\omega)^2\vartheta\big({-}p^{1/2}\big)\vartheta\big(p^{1/2}\omega\big)^3}\bigg)^2
=(2 \psi(\psi+1)(2\psi+1))^2,
\\[.5ex]
\bigg(\frac{\omega^2\vartheta\big({-}p^{1/2}\omega\big)\vartheta\big({-}p^{1/2}\big)\vartheta(\omega)^2} {\vartheta\big(p^{1/2}\omega\big)^2\vartheta(-1)^2}\bigg)^3
=\frac{(\psi+1)(2\psi+1)^3}{2\psi^5},
\\[.5ex]
\label{eq:E}
\frac{1}{\xi}\frac{\omega^2\vartheta\big({-}p^{1/2}\omega\big)\vartheta\big({-}p^{1/2}\big)\vartheta(\omega)^2} {\vartheta\big(p^{1/2}\omega\big)^2\vartheta(-1)^2}
=\frac{2\psi+1}{\psi},
\end{gather}
and
\begin{gather*}
\xi^3=\frac{\psi+1}{2\psi^2}.
\end{gather*}
Insert \eqref{eq:n01}--\eqref{eq:E} into \eqref{partwithrhominus1}. We need to separate the different cases for $n\operatorname{mod} 3$, since
\begin{gather*}
2n^2+4n+b \equiv
\begin{cases}
1 \mod 3, &\text{for\ $n\equiv 0,1 \mod 3$},\\
0 \mod 3, &\text{for\ $n\equiv 2 \mod 3$},
\end{cases}
\end{gather*}
and
\begin{gather*}
n^2-n\equiv
\begin{cases}
0 \mod 6, &\text{for\ $n\equiv 0,1 \mod 3$},\\
2 \mod 6, &\text{for\ $n\equiv 2 \mod 3$}.
\end{cases}
\end{gather*}
Although the computations are slightly different,
\begin{align*}
E^{n^2-n}D^{n-1}
=\xi^{2n^2-n-1}\left(\frac{\psi}{2\psi+1}\right)^{n^2-2n+1}(\psi+1)^{2n-2n^2},
\end{align*}
for all $n$, where we have used \eqref{lemma9.1-2psi+1} and \eqref{lemma9.1-psi+1}.

The equation then finally becomes
\begin{gather}
\binom{n-1}{m-1} T(2\psi+1, \dots, 2\psi+1, \psi)\label{finalpartitionfcn}
\\ \qquad
{}=\bigg(\frac{2\psi\!+\!1}{\psi}\bigg)^{n^2-2n+1}\!\!\!(\psi\!+\!1)^{2n^2-2n}
\sum_{\mathclap{\substack{l\equiv n \operatorname{mod} 3 \\ k_0\in\Z}}} (-1)^{n+l}(N_{m,l}(k_0)\!+\!N_{m,l-1}(k_0))\xi^{3k_0-4n^2-3n+l-c},
\nonumber
\end{gather}
for\begin{gather*}
c=\begin{cases}
3, & n\equiv 0,1 \mod 3,\\
1, &n\equiv 2 \mod 3.
\end{cases}
\end{gather*}
Likewise if we put $\rho=-\omega$ into \eqref{part2} we get
\begin{gather}
\binom{n-1}{m} T(2\psi+1, \dots, 2\psi+1, \psi)\label{finalpartitionfcn2}
\\ \qquad
{}=\!\bigg(\frac{2\psi\!+\!1}{\psi}\bigg)^{n^2-2n+1}\!\!\!\!\!(\psi\!+\!1)^{2n^2-2n}
\sum_{\mathclap{\substack{l\equiv n \operatorname{mod} 3 \\ k_2\in\Z}}} (-1)^{n+l}(N_{m,l+1}(k_2)\!+\!N_{m,l+2}(k_2))\xi^{3k_2-4n^2-3m+l-d},
\nonumber
\end{gather}
and putting $\rho=-\omega^2$ into \eqref{part2} yields
\begin{gather}
\binom{n}{m} T(2\psi+1, \dots, 2\psi+1, \psi)
\label{finalpartitionfcn3}
\\ \qquad
{}=\!\bigg(\frac{2\psi\!+\!1}{\psi}\bigg)^{n^2-2n+1}\!\!\!\!\!\!\!(\psi\!+\!1)^{2n^2-2n} \!\sum_{\mathclap{\substack{l\equiv n \operatorname{mod} 3 \\ k_1\in\Z}}} (-1)^{n+l}(N_{m,l}(k_1)\!+\!N_{m,l+1}(k_1))\xi^{3k_1-4n^2-3n+3m+l-d},
\nonumber
\end{gather}
for
\begin{gather*}
d=\begin{cases}
0, & n\equiv 0,1 \mod 3,\\
1, & n\equiv 2 \mod 3.
\end{cases}
\end{gather*}

\section[Formulas for p\_n(z)]{Formulas for $\boldsymbol{p_n(z)}$}\label{sec:pnintermsofT}
Bazhanov and Mangazeev \cite{BazhanovMangazeev2005} introduced certain polynomials, describing the ground state eigenvalue of Baxter's $Q$-operator \cite{Baxter1972} for the 8V model in the case with $\eta=-2/3$. These polynomials are given by
\begin{gather*}
\mathcal{P}_n(x, z)=\sum_{k=0}^n r_k^{(n)}(z) x^k,
\end{gather*}
normalized by $r_n^{(n)}(0)=1$. Furthermore, they introduced $s_n(z)=r_n^{(n)}(z)$ and $\overline{s}_n(z)=r_0^{(n)}(z)$.
In a later paper \cite{BazhanovMangazeev2010}, they connected these polynomials to the ground state eigenvectors of the supersymmetric XYZ-Hamiltonian for spin chains of odd length $2n+1$. Several conjectures about the polynomials were stated, including that $s_n(z)$ factors into certain polynomials which seem to have positive coefficients. Here, polynomials $p_n(z)$ and $q_n(z)$, with $\deg p_n(z)=\deg q_n(z)=n(n+1)$ and $p_n(0)=q_n(0)=1$, show up as factors of $s_{2n}\big(z^2\big)$. In further conjectures, they suggest that certain components of the ground state eigenvectors for the supersymmetric XYZ spin chain can be written in terms of $s_n(z)$, $\overline{s}_n(z)$, $p_n(z)$ and $q_n(z)$. In \cite{Hietala2020}, we investigated the polynomials $q_n(z)$. In this paper, we turn the focus to the polynomials $p_n(z)$.

Zinn-Justin \cite{Zinn-Justin2013} investigated a family of polynomials equivalent to Rosengren's polynomials~$T$. He realized that certain specializations of the parameters in his polynomials seem to yield the polynomials of Bazhanov and Mangazeev.
Let
\cite{Rosengren2015}
\begin{gather*}
\xi_0=2\psi+1, \qquad \xi_1=\frac{\psi}{\psi+2}, \qquad \xi_2=\frac{\psi(2\psi+1)}{\psi+2}, \qquad \xi_3=1.
\end{gather*}
In the case where all parameters $x_i$ of Rosengren's polynomials $T$ \eqref{T} are one of the $\xi_j$'s, the only dependence on a variable is that the $\xi_j$'s depend on $\psi$. Therefore we define
\begin{gather*}
t^{(k_0, k_1, k_2, k_3)}(\psi)\coloneqq T(\underbrace{\xi_0, \dots, \xi_0}_{k_0 \text{ times}}, \underbrace{\xi_1, \dots, \xi_1}_{k_1 \text{ times}}, \underbrace{\xi_2, \dots, \xi_2}_{k_2 \text{ times}}, \underbrace{\xi_3, \dots, \xi_3}_{k_3 \text{ times}}),
\end{gather*}
with $k_j$ non-negative integers and $k_0+k_1+k_2+k_3=2n$.
In \cite{Rosengren2015}, this definition is extended to negative values of $k_i$.
In terms of Rosengren's polynomials, the polynomials $p_n$ are, up to a~pre\-factor, equivalent to $t^{(2n+1, 0,0,-1)}(\psi)$,
which in turn corresponds to $T(\underbrace{2\psi+1, \dots, 2\psi+1}_{2n+1 \text{ times}}, \psi)$. The second correspondence can be seen by using \cite[formulas~(2.12) and (2.13)]{Rosengren2015}. This together with the expression in \cite[Section~5.3]{Rosengren2015} and the symmetries in \cite[Proposition 2.2]{Rosengren2015} yield
\begin{gather}
\label{def:pn}
T(\underbrace{2\psi+1, \dots, 2\psi+1}_{2n-1 \text{ times}}, \psi)
=\bigg(\frac{\psi}{2\psi+1}\bigg)^{n-1}\big((\psi+1)(2\psi+1)^2\big)^{n^2-n}p_{n-1}\bigg({-}\frac{1}{2\psi+1}\bigg).
\end{gather}
That these polynomials are equivalent to the polynomials of Mangazeev and Bazhanov is still a~conjecture, but here we take \eqref{def:pn} as the definition of $p_n$.

We insert \eqref{def:pn} into \eqref{finalpartitionfcn} and get
\begin{gather*}
\binom{n-1}{m-1} p_{n-1}\bigg({-}\frac{1}{2\psi+1}\bigg)
\\ \qquad
{}=\bigg(\frac{\psi+1}{\psi(2\psi+1)}\bigg)^{n^2-n}
\sum_{\mathclap{\substack{l\equiv n \operatorname{mod} 3 \\ k_0\in\Z}}} (-1)^{n+l}(N_{m,l}(k_0)+N_{m,l-1}(k_0))\xi^{3k_0-4n^2-3n+l-c},
\end{gather*}
where
\begin{gather*}
c=\begin{cases}
3, & n\equiv 0,1 \mod 3,\\
1, &n\equiv 2 \mod 3.
\end{cases}
\end{gather*}
We can do the same substitution in \eqref{finalpartitionfcn2} and \eqref{finalpartitionfcn3}.
Finally, changing to the variable $z=-1/(2\psi+1)$ yields the following theorem.
\begin{Theorem}
\label{mainthm}
Let $N_{m,l}(k_i)$ be the number of states with $m$ positive turns, $k_i$ faces of color $i$, and where the left arrow of the second to last column is on the $l$th row from below.
Formulas for $p_{n-1}(z)$ are then given by
\begin{gather*}
\binom{n-1}{m-1} p_{n-1}(z) =
\sum_{\mathclap{\substack{l\equiv n \operatorname{mod} 3 \\ k_0\in\Z}}} (-1)^{n+l}(N_{m,l}(k_0)+N_{m,l-1}(k_0))\frac{(z(z-1))^{(3k_0-n^2-6n+l-c)/3}}{(z+1)^{(6k_0-5n^2-9n+2l-2c)/3}},
\end{gather*}
where
\begin{gather*}
c=\begin{cases}
3, & n\equiv 0,1 \mod 3,\\
1, &n\equiv 2 \mod 3,
\end{cases}
\\[1ex]
\binom{n}{m} p_{n-1}(z)
=\sum_{\mathclap{\substack{l\equiv n \operatorname{mod} 3 \\ k_1\in\Z}}} (-1)^{n+l}(N_{m,l}(k_1)+N_{m,l+1}(k_1))\frac{(z(z-1))^{(3k_1-n^2-6n+3m+l-d)/3}}{(z+1)^{(6k_1-5n^2-9n+6m+2l-2d)/3}},
\end{gather*}
and
\begin{gather*}
\binom{n\!-\!1}{m} p_{n-1}(z)
=\sum_{\mathclap{\substack{l\equiv n \operatorname{mod} 3 \\ k_2\in\Z}}} (-1)^{n+l}(N_{m,l+1}(k_2)\!+\!N_{m,l+2}(k_2)) \frac{(z(z\!-\!1))^{(3k_2-n^2-3n-3m+l-d)/3}}{(z\!+\!1)^{(6k_2-5n^2-3n-6m+2l-2d)/3}},
\end{gather*}
where
\begin{gather*}
d=\begin{cases}
0, & n\equiv 0,1 \mod 3,\\
1, & n\equiv 2 \mod 3.
\end{cases}
\end{gather*}
\end{Theorem}
The polynomials $p_n(z)$ are known to have the symmetry \cite{BazhanovMangazeev2010}
\begin{gather*}
p_n(z)=\bigg(\frac{1+3z}{2}\bigg)^{n(n+1)} p_n\bigg(\frac{1-z}{1+3z}\bigg),
\end{gather*}
and it is immediate that this holds for the expressions in Theorem~\ref{mainthm}.

Theorem~\ref{mainthm} gives three expressions for $p_{n-1}(z)$ in terms of three-colorings. The functions~$p_{n-1}(z)$ are known to be polynomials, so this must of course hold for the right hand sides of the expressions as well, even though the individual terms in the sums may be rational functions, with possible poles in $z=-1, 0, 1$. Any terms with poles and nonzero coefficients have to cancel.
In the corresponding expression for $q_{n-1}(z)$, the individual terms can not contain poles, and this let us, among other things, draw conclusions about the possible values for the numbers $k_i$ \cite{Hietala2020}. Because of the factors $(-1)^{n+l}$ in the expressions for $p_{n-1}(z)$ in Theorem~\ref{mainthm}, similar conclusions cannot be drawn in this case.

Although we have found combinatorial expressions for $p_{n-1}(z)$, we can not immediately see from the formulas that the coefficients are positive. To draw conclusions about the coefficients, one would need to know something more about the numbers $N_{m,l}(k_i)$.

\subsection*{Acknowledgements}

This paper was written during my time as a Ph.D.~student at the Department of Mathematical Sciences, University of Gothenburg and Chalmers University of Technology.
I am very thankful to my supervisor Hjalmar Rosengren and cosupervisor Jules Lamers for a lot of encouragement and support with my research. I also would like to thank the anonymous referees for nice and useful comments.

\pdfbookmark[1]{References}{ref}
\LastPageEnding

\end{document}